\documentclass[journal]{IEEEtran}
\usepackage{amsmath,dsfont,bbm,epsfig,amssymb,amsfonts,amstext,verbatim,amsopn,cite,subfigure}
\usepackage{multirow,multicol,lipsum,xfrac}
\usepackage{amsthm}
\usepackage{mathtools,amsthm}
\usepackage{url}
\usepackage{amsfonts}
\usepackage{epsfig}
\usepackage[font=small]{caption}
\usepackage{psfrag}	
\usepackage{algorithm,algorithmic}
\usepackage{pifont}
\usepackage[utf8]{inputenc}
\usepackage[T1]{fontenc}  
\usepackage[nolist]{acronym}
\usepackage{paralist}
\usepackage{enumitem}
\usepackage{bbm}
\usepackage[process=auto]{pstool}
\usepackage{float}
\usepackage{subfigure, xfp}
\usepackage{graphicx}
\usepackage{caption}
\usepackage{tikz,pgfplots}
\usetikzlibrary{shapes.misc,arrows,positioning}

\usepackage[utf8]{inputenc} 
\usepackage{amsmath,dsfont,bbm,epsfig,amssymb,amsfonts,amstext,verbatim,amsopn,cite}
\usepackage{multirow,multicol,lipsum,xfrac}
\usepackage{amsthm}
\usepackage{mathtools,amsthm}
\usepackage{url}
\usepackage{amsfonts}
\usepackage{tikz,pgfplots}
\usepackage[nolist]{acronym}

\newcommand{\setR}{\mathbbmss{R}}

\newcommand{\setS}{\mathbbmss{S}}
\newcommand{\setW}{\mathbbmss{W}}

\newcommand{\setC}{\mathbbmss{C}}

\newcommand{\Ex}[2]{ \mathbbm{E}_{#2} \left\lbrace #1 \right\rbrace }

\newcommand{\her}{\mathsf{H}}

\newcommand{\argmax}{\mathop{\mathrm{argmax}}}

\newcommand{\maK}{\mathcal{K}}

\newcommand{\maP}{\mathcal{P}}
\newcommand{\mae}{\mathcal{E}}

\newcommand{\mas}{\mathcal{S}}

\newcommand{\mam}{\mathcal{M}}
\newcommand{\maD}{\mathcal{D}}

\newcommand{\mal}{\mathcal{L}}

\newcommand{\bxx}{\mathbf{x}}

\newcommand{\bww}{\mathbf{w}}
\newcommand{\buu}{\mathbf{u}}

\newcommand{\bvv}{\mathbf{v}}
\newcommand{\byy}{\mathbf{y}}
\newcommand{\bff}{\mathbf{f}}

\newcommand{\rmr}{\mathrm{r}}

\newcommand{\bh}{{\mathbf{h}}}

\newcommand{\bphi}{\boldsymbol{\phi}}

\newcommand{\set}[1]{\left\lbrace#1\right\rbrace}
\newcommand{\Diag}[1]{\mathrm{Diag}\left\lbrace #1 \right\rbrace}
\newcommand{\brc}[1]{\left( #1 \right) }

\newcommand{\dbc}[1]{\left[ #1 \right] }

\newcommand{\br}{{\mathbf{r}}}

\newcommand{\bc}{{\mathbf{c}}}

\newcommand{\bzz}{{\mathbf{z}}}
\newcommand{\baa}{{\mathbf{a}}}

\newcommand{\trp}{\mathsf{T}}

\newcommand{\mA}{\mathbf{A}}
\newcommand{\mC}{\mathbf{C}}
\newcommand{\mF}{\mathbf{F}}
\newcommand{\mR}{\mathbf{R}}

\newcommand{\mI}{\mathbf{I}}

\newcommand{\mG}{\mathbf{G}}
\newcommand{\mQ}{\mathbf{Q}}

\newcommand{\mU}{\mathbf{U}}

\newcommand{\mD}{\mathbf{D}}

\newcommand{\mX}{\mathbf{X}}

\newcommand{\mW}{{\mathbf{W}}}

\newcommand{\mB}{\mathbf{B}}
\newcommand{\mT}{\mathbf{T}}
\newcommand{\mH}{\mathbf{H}}

\newcommand{\rme}{\mathrm{e}}

\newcommand{\itr}[1]{^{\left( #1 \right)} }

\newcommand{\norm}[1]{\lVert #1 \rVert}

\newcommand{\abs}[1]{\lvert #1 \rvert}
\newcommand{\tr}[1]{\mathrm{tr} \{ #1 \}}

\newtheoremstyle{mystyle}
{}
{}
{\it}
{}
{\bfseries}
{:}
{ }
{}

\theoremstyle{mystyle}

\newtheorem{definition}{Definition}
\newtheorem{proposition}{Proposition}
\newtheorem{remark}{Remark}

\newtheorem{lemma}{Lemma}

\newcounter{bar}

\usepackage{pgfplots}
\usepackage{filecontents}
\usetikzlibrary{calc}

\begin{document}
	
	\newcommand{\xmark}{\ding{55}}%
	\newcommand{\bomega}{\boldsymbol{\omega}}
	\begin{acronym}
		\acro{ai}[AI]{artificial intelligence}
		\acro{Cram\'er-Rao bound}[Cram\'er-Rao bound]{Cram\'er-Rao bound}
		\acro{mimo}[MIMO]{multiple-input multiple-output}
		\acro{miso}[MISO]{multiple-input single-output}
		\acro{csi}[CSI]{channel state information}
		\acro{ai}[AI]{artificial intelligent}
		\acro{awgn}[AWGN]{additive white Gaussian noise}
		\acro{iid}[i.i.d.]{independent and identically distributed}
		\acro{uts}[UTs]{user terminals}
		\acro{ps}[PS]{parameter server}
		\acro{irs}[IRS]{intelligent reflecting surface}
		\acro{tas}[TAS]{transmit antenna selection}
		\acro{isac}[ISAC]{integrated sensing and communication}
		\acro{rhs}[r.h.s.]{right hand side}
		\acro{lhs}[l.h.s.]{left hand side}
		\acro{wrt}[w.r.t.]{with respect to}
		\acro{rs}[RS]{replica symmetry}
		\acro{mac}[MAC]{multiple access channel}
		\acro{np}[NP]{non-deterministic polynomial-time}
		\acro{papr}[PAPR]{peak-to-average power ratio}
		\acro{rzf}[RZF]{regularized zero forcing}
		\acro{snr}[SNR]{signal-to-noise ratio}
		\acro{sinr}[SINR]{signal-to-interference-and-noise ratio}
		\acro{svd}[SVD]{singular value decomposition}
		\acro{mf}[MF]{matched filtering}
		\acro{gamp}[GAMP]{generalized AMP}
		\acro{amp}[AMP]{approximate message passing}
		\acro{vamp}[VAMP]{vector AMP}
		\acro{map}[MAP]{maximum-a-posterior}
		\acro{ml}[ML]{maximum likelihood}
		\acro{mse}[MSE]{mean squared error}
		\acro{mmse}[MMSE]{minimum mean squared error}
		\acro{ap}[AP]{average power}
		\acro{ldgm}[LDGM]{low density generator matrix}
		\acro{tdd}[TDD]{time division duplexing}
		\acro{rss}[RSS]{residual sum of squares}
		\acro{rls}[RLS]{regularized least-squares}
		\acro{ls}[LS]{least-squares}
		\acro{erp}[ERP]{encryption redundancy parameter}
		\acro{zf}[ZF]{zero forcing}
		\acro{ta}[TA]{transmit-array}
		\acro{ofdm}[OFDM]{orthogonal frequency division multiplexing}
		\acro{dc}[DC]{difference of convex}
		\acro{bcd}[BCD]{block coordinate descent}
		\acro{mm}[MM]{majorization-maximization}
		\acro{bs}[BS]{base-station}
		\acro{ota}[OTA]{over-the-air}
		\acro{ULA}[ULA]{uniform linear array}
		\acro{fl}[FEEL]{federated edge learning}
		\acro{ota-fl}[OTA-FEEL]{over-the-air federated learning}
		\acro{los}[LoS]{line-of-sight}
		\acro{nlos}[NLoS]{non-line-of-sight}
		\acro{AoA}[AoA]{angle of arrival}
		\acro{AoD}[AoD]{angle of departure}
		\acro{CNN}[CNN]{convolutional neural network}
		\acro{sgd}[SGD]{stochastic gradient descent}
		\acro{aircomp}[AirComp]{over-the-air computation}
		\acro{lmi}[LMI]{linear matrix inequalities}
		\acro{iot}[IoT]{internet of things }
		\acro{iscco}[ISCCO]{Integrated Sensing, Communication and Computation}
		\acro{nn}[NN]{neural network}
		\acro{cnn}[CNN]{convolutional neural network}
		\acro{oac}[OAC]{over-the-air computation}
		\acro{iscc}[ISCC]{integrated sensing communication and computation}
		\acro{mec}[MEC]{mobile edge computing}
	\end{acronym}
	
	\title{Over-the-Air~FEEL with Integrated Sensing: Joint Scheduling and Beamforming Design}
	
	\author{
		\IEEEauthorblockN{
			Saba Asaad, \textit{Member IEEE}, 
			Ping Wang, \textit{Fellow IEEE},
            and Hina Tabassum, \textit{Senior Member IEEE}\vspace{-5mm}
			\thanks{This research was supported by NSERC Discovery Grant funded by the
				Natural Sciences and Engineering Research Council of Canada (NSERC) and Banting Fellowship. 
				\\
				Saba Asaad, Ping Wang, and Hina Tabassum are with the Department of Electrical Engineering and Computer Science at York University, Toronto, Canada; emails:  \{asaads, pingw, hinat\}@yorku.ca.}
		}
	}
	
	\IEEEoverridecommandlockouts
	\maketitle
	\begin{abstract}
\raggedbottom 
Employing wireless systems with dual sensing and communications functionalities is becoming critical in next generation of wireless networks. 
In this paper, we propose a robust design for over-the-air federated edge learning (OTA-FEEL) that leverages sensing capabilities at the \ac{ps} to mitigate the impact of target echoes on the analog model aggregation. We first derive novel expressions for the Cram\'er-Rao bound of the target response and \ac{mse} of the estimated global model to measure radar sensing and model aggregation quality, respectively. Then, we develop a joint scheduling and beamforming framework that optimizes the OTA-FEEL performance while keeping the sensing and communication quality, determined respectively in terms of Cram\'er-Rao bound and achievable downlink rate, in a desired range. The resulting scheduling problem reduces to a combinatorial mixed-integer nonlinear programming problem (MINLP). We develop a low-complexity hierarchical method based on the matching pursuit algorithm used widely for sparse recovery in the literature of compressed sensing. The proposed algorithm uses a step-wise strategy to omit the least effective devices in each iteration based on a metric that captures both the aggregation and sensing quality of the system. It further invokes alternating optimization scheme to iteratively update the downlink beamforming and uplink post-processing by marginally optimizing them in each iteration. Convergence and complexity analysis of the proposed algorithm is presented. Numerical evaluations on MNIST and CIFAR-10 datasets demonstrate the effectiveness of our proposed algorithm. The results show that by leveraging accurate sensing, the target echoes on the uplink signal can be effectively suppressed, ensuring the quality of model aggregation to remain intact despite the interference. 
	\end{abstract}	
\vspace{-0.2cm}
\section{Introduction}

Next generations of data networks are expected to support a wide range of artificial intelligence (AI) services \cite{zhu2023pushing}. While conventional communication networks exchange exact information context, e.g., text and voice, the recent data networks exchange parameters that \textit{implicitly} convey information to provide different functionalities across the network \cite{yao2019artificial}. Examples of such functionalities are \ac{isac} and \ac{ota-fl}. 
In \ac{isac}, the sensing data is collected from the already-existing communication signals in the form of echoes and extracted using post-processing techniques to execute a desired sensing task such as localization or velocity estimation.
Similarly, OTA-FEEL uses the physical superposition of the transmitted signals over the uplink wireless channels to realize the intended model aggregation of the FEEL framework over the air. The exchanged data over the network contains the parameters of the local and global models, e.g., the gradients of losses. These parameters are utilized to complete an specific learning task, e.g., to update model weights via \ac{sgd}, and hence \textit{only implicitly} impact the final accuracy of the trained model. 

While integrating multiple functionalities enhances wireless network efficiency, it can also increase vulnerability to cross-functionality couplings, i.e., the undesired interference introduced to the system due to their multi-purpose design. Such couplings, which are often unavoidable, can potentially degrade the expected performance of the system. 
 \subsection{Background and Related Work}
In OTA-FEEL, the server directly recover a noisy aggregation of model updates by exploiting the superposition property of wireless channels. This approach eliminates the need to decode each device’s update individually, allowing simultaneous transmission over the same channel. As a result, OTA-FEEL becomes significantly more communication-efficient, especially with a larger number of participating devices.
However, the analog nature of OTA-FEEL introduces challenges such as vulnerability to fading, noise, and interference \cite{amiri2020federated, asaad2024joint}. These imperfections impact the training process and the effectiveness of model convergence.
As the direct result to this property, efficient signal post-processing within the analog domain is crucial in \ac{ota-fl}.
At each edge device, precoding coefficients are applied before transmission to reduce the effects of channel issues like interference and fading. This “pre-distortion” adjusts each signal so that, after passing through the imperfect channel, the combined signals resemble an accurate average of the local models.
At the \ac{ps}, the receive beamformer is designed to further correct or “post-process” the combined signal after it is received through the noisy channel. This ensures that, despite interference or noise in transmission, the aggregated signal closely approximates the target function.

The initial work \cite{yang2020federated} focused on studying the joint design of beamforming and device scheduling policy in MIMO settings. Building upon that, \cite{yang2022over} investigated the joint adaptive local computation and power control for OTA-FEEL. In \cite{wang2022edge}, a novel unit-modulus computation framework was proposed to minimize communication delay and implementation costs through analog beamforming. Additionally, \cite{sedaghat2023novel} proposed low-complexity algorithms for device coordination in OTA-FEEL using the minimum mean squared error (MMSE) and zero forcing (ZF) methods. The impact of hardware impairments, including low-resolution digital-to-analog converters (DACs), phase noise, and non-linearities in power amplifiers, is examined in \cite{tegin2021blind}, where it is shown that these factors have a negligible effect on the convergence and accuracy of FEEL. 
\par \Ac{isac} refers to a unified system in which wireless communication and radio sensing are integrated within a single architecture with common signal-processing modules. This integration allows \ac{isac} to pursue mutual benefits and significantly enhance spectrum, energy, and hardware efficiency as compared to separate communication and radar sensing in isolated frequency bands \cite{liu2022survey}. 
 \color{black}
Most of the recent studies considered ISAC in a distributed manner where the edge devices act as sensors that first send signals to sense the targets and then deliver the information to the server. 
In \cite{li2022integratedconf}, the authors studied multiple multi-antenna devices with separate antennas for transmission and target echo reception. The target is assumed to be far from the server, ensuring no interference between the target reflections and the signal received by the parameter server (PS) over the multiple access channel. The optimization objective is to minimize the \ac{mse} of target estimation and function computation, resulting in a joint design of transmit beamformers at edge sensors and a data aggregation beamformer at the PS.
Building upon this, the authors in \cite{li2023integrated2} extended  the optimization framework to include the case where transmit antennas at the sensors are divided between radar sensing and OTA computation tasks. 
The problem of client scheduling in \ac{irs}-assisted OTA-FEEL with ISAC has been studied in \cite{zheng2023federated}. The objective is to maximize the number of participating users while meeting the \ac{mse} constraints for both communication and sensing. In this work, radar echo sensing at the devices remains interference-free from downlink PS communication, without any communication rounds.

Previous research on OTA-FEEL has primarily concentrated on classical information-conveying communication systems, leaving the co-existence of sensing largely unexplored and not considering the significant interference that can severely impact analog computation. ISAC in FEEL has recently been explored in \cite{liu2022toward, liu2022vertical}, incorporating sensing capabilities at the devices. For instance, in \cite{liu2022toward}, the authors addressed resource allocation in an \ac{isac}-based FEEL scenario. The approach involved multiple \ac{isac} devices collecting local datasets through wireless sensing in one time slot, and subsequently exchanging model updates with a central server over wireless channels in another time slot. 
In \cite{xing2023task}, a wireless FEEL system with multiple ISAC devices is considered, where clients acquire their specific datasets for human motion recognition through wireless sensing and subsequently communicate exclusively with the edge server to exchange model updates.
The authors  in \cite{liu2022vertical} proposed vertical FEEL where, unlike horizontal FEEL, each edge device has a distinct transmission bandwidth for ISAC purposes, and a coordinating  device performs the target sensing. It is worth mentioning that the aforementioned studies do not take into account the interference of sensing reflections on the analog computation at the PS. This follows from the fact that sensing and OTA aggregation, thanks to the large distance between the target and PS, are orthogonal in time. Although this assumption is realistic in an application-specific remote sensing system, in practical wireless networks it is hardly the case. In fact, in these systems, the received signal at the PS can be severely interfered by reflections subject to sensing.

\raggedbottom
\subsection{Motivation and Contributions}
The core contribution of this work is to 
propose a robust design for OTA-FEEL in the presence of sensing where reflections from the targets can severely degrade the performance of analog model aggregation involved in  over-the-air computation. The motivation of this work arises from recent wireless system advancements enabling joint sensing and communications functionalities. These systems often rely on target reflections of the dual-purpose signals, and hence operate in a mode where target echoes can create interference to communication signals or model aggregation in OTA-FEEL which is our case. Unlike digital communications, which can handle such interference by proper coding techniques, the OTA-FEEL operates in the analog domain and hence is very vulnerable to interference. Our study proposes a novel scheme based on successive interference cancellation to handle this issue in the analog domain. 
The key difference of our study with the existing work on successive interference cancellation that is used in non-orthogonal multiple access integrated (NOMA)-ISAC systems, e.g., \cite{ouyang2023revealing}, is the domain in which the cancellation is being performed. In NOMA-ISAC, the signaling is being performed in the digital domain, whereas in our setting the echo interference is superimposed by uplink signals that are used for aggregation in the analog model. Moreover, standard interference mitigation algorithms generally focus on suppressing interference to improve the quality of model \cite{ouyang2022performance}. However, in our framework, we consider incorporating echo interference cancellation as part of a successive cancellation strategy tailored for \textit{joint sensing and model aggregation}. This approach is more practical for our setting because it {not only} enhances aggregation quality but also utilizes interference as a useful signal for target estimation, providing a dual-purpose benefit that is not available in standard methods. 

Motivated by the above, we consider the problem of \ac{ota-fl} where the \ac{ps} is receiving echoes of its downlink transmission from the scattered targets in the network. Due to difference in time delay, these reflections interfere the analog model aggregation in the uplink leading to large aggregation error. In the sequel, we propose a joint design in which the sensing is \textit{not simply implemented as a new functionality, but rather as a necessity to control its interference on model aggregation.} We characterize the trade-off between sensing and aggregation which unlike the sensing-communication, and aggregation-communication trade-offs is not very evident. On one hand, with weak echoes, the sensing performance is compromised, leading to a higher sensing error with limited impact on the analog aggregation due to its low-power share. On the other hand, with substantial interference, the sensing error is reduced significantly. The impact of this small error on model aggregation  is however amplified  due to its high-power share in uplink. It is thus not immediately clear which regime is optimal and how the aggregation process behaves as the power of echoes varies. We investigate this trade-off in detail to shed light on these intricacies. 

  Our contributions is summarized below:
	\begin{itemize}
		    \item We develop a joint scheduling and beamforming framework to optimize the performance of OTA-FEEL in the presence of  sensing at the \ac{ps}. The framework  enables  suppressing the echo interference on the computation signal by estimating the target parameters. To mitigate the interference caused by target echoes on the aggregated model, we develop a successive interference cancellation scheme which estimates the target parameters and uses  estimated parameters to cancel out the impact of target reflections on the analog model aggregation.
    \item We derive novel expressions for the model aggregation error and Cram\'er-Rao bound on the sensing error for successive interference cancellation. Furthermore, we use these analytic characterizations to formulate the device scheduling and beamforming problem.
    \item We consider joint beamforming and device scheduling problem in which the PS schedules maximal active devices while ensuring the sensing and model aggregation errors remain below a desired level, and that the minimum communication requirements for downlink model broadcast are met. The trade-offs between sensing accuracy and learning quality are then characterized. 
  \item  Leveraging on tools from matching pursuit and optimization, we develop a low-complexity  algorithm to  solve the combinatorial mixed integer non-linear programming problem (MINLP). The algorithm iteratively identifies the largest set of active edge devices using matching pursuit and alternates between downlink beamforming and uplink post-processing. Notably, we demonstrate that these marginal problems decouple in each iteration, resulting in a substantial reduction in computational complexity. Additionally, we provide complexity and convergence analysis for the proposed algorithm.
  \item  
  We conduct extensive numerical investigations on MNIST and CIFAR-10 datasets. Numerical results depict the performance of the proposed algorithm compared to the existing benchmarks.  Our numerical results demonstrate that in various operating points of the system, interference suppression comes as an immediate result of sensing. Due to the accurate sensing, the target echoes on the uplink signal can be effectively cancelled out ensuring that the quality of model aggregation remains intact.  
	\end{itemize}
	\vspace{-0.3cm}
	\subsection{Paper Organization and Notations}
The paper's structure is as follows: Section II introduces the system model. Section III characterizes performance metrics to evaluate sensing and learning quality. Section IV discusses the problem formulation, and Section V presents the algorithm design using the matching pursuit technique. Section VI showcases simulation results, and Section VII concludes the paper.
	
	Scalars, vectors, and matrices are denoted by non-bold, bold lower-case, and bold upper-case letters, respectively. The transpose , the conjucate and the conjugate transpose of $\mH$ are denoted by $\mH^{\trp}$, $\mH^*$ and $\mH^{\her}$, respectively. $\mI_N$ denotes the $N\times N$ identity matrix. The sets $\setR$ and $\setC$ are the real axis and the complex plane, respectively. $\mathcal{CN}\brc{\eta,\sigma^2}$ is the complex Gaussian distribution with mean $\eta$ and variance $\sigma^2$. $\dagger$ denotes the pseudo-inverse. The symbol $\otimes$ is the Kronecker product. $\Ex{.}{}$ denotes the expectation. To keep it concise, $\set{1,\ldots,N}$ is abbreviated as $\dbc{N}$. The symbol $\odot$ is the entry-wise product and  $\nabla$ denotes the gradient operator.
    \vspace{-0.3cm}
	\section{System Model and Assumptions}
	We consider a dual-function \ac{ps} equipped with $N$ antenna elements to support OTA-FEEL with integrated sensing as presented in Fig. \ref{Fig1_Schematic}. A set of $\mathcal{K}={1, 2, \cdots, K}$ single-antenna edge devices and a target located around the \ac{ps} in an open area at a moderate distance and a certain \ac{AoA} are considered. \color{black} The \ac{ps} performs sensing periodically within certain blocks using the echoes of its downlink transmissions. It is assumed that the uplink and downlink transmissions are performed in the same band, and hence the echoes of downlink transmission can interfere with the uplink transmissions of the devices. We denote the length of the sensing blocks with $L$ and refer to the index of a given interval as $\ell \in \set{1,\ldots,L}$. This implies that starting from interval $\ell=1$, the \ac{ps} collects a block of $L$ samples and estimates the target unknowns, i.e., target location, at the end of the block. The block-length $L$ is assumed to be smaller than the channel coherence time, i.e., $ L < L_{\rm coh}$, with $L_{\rm coh}$ being the minimum of uplink and downlink coherence times. Moreover, We assume that the block length $L$ is smaller than the downlink transmission time, i.e., $M$-symbol duration.
The \ac{ps} aggregates a global model by leveraging the over-the-air transmission of locally trained models from the edge devices, a hallmark of OTA-FL. Simultaneously, the PS exploits the echoes from its downlink transmission of the aggregated model to sense the target. In the following, we describe each stage of the OTA-FEEL process with integrated sensing.
 \color{black}

	\begin{figure}
		\centering
		\includegraphics[scale=.25]{ 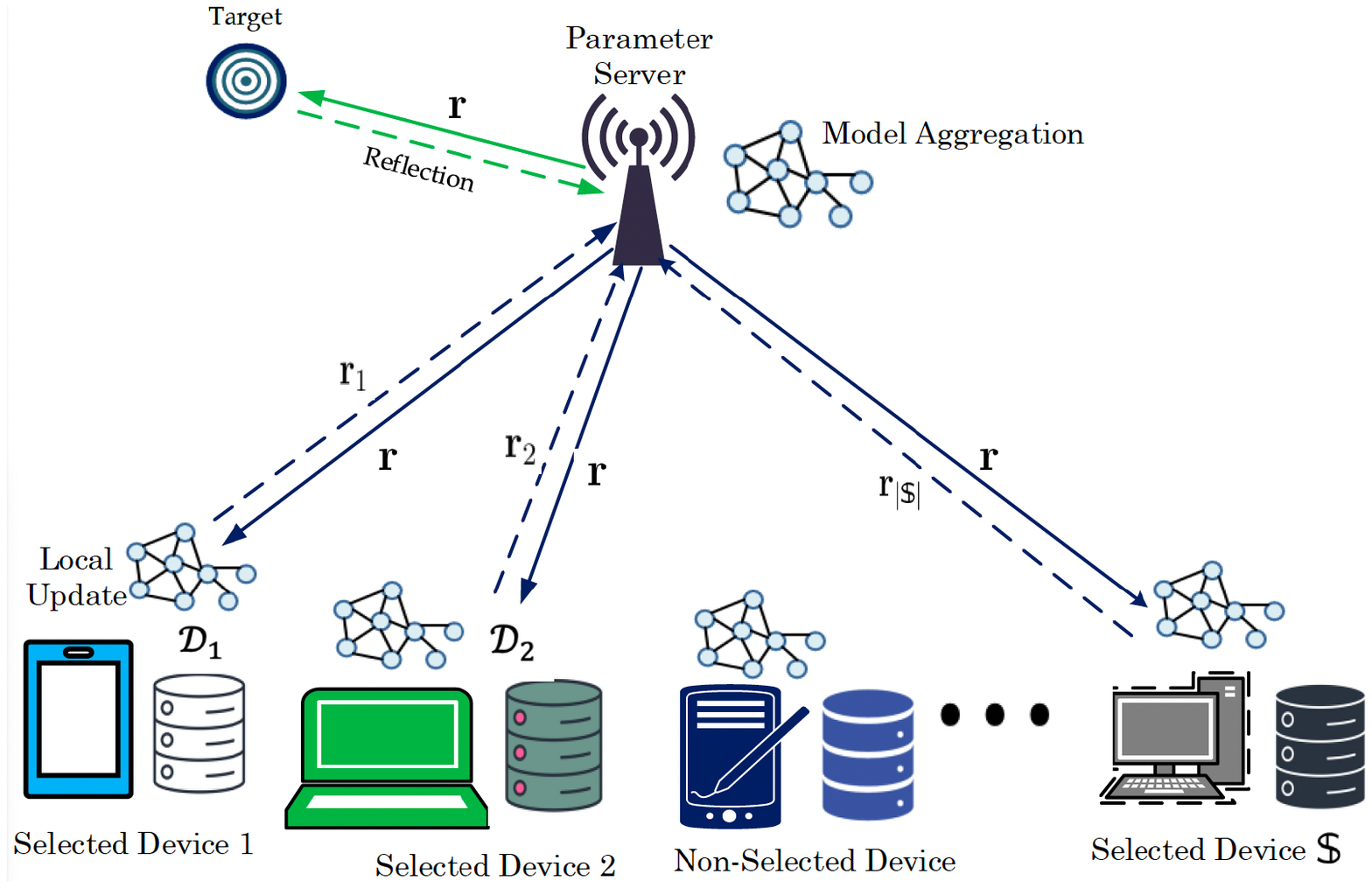}
		\caption{Illustration of the integrated OTA-FEEL and sensing setting. 
  }
		\label{Fig1_Schematic}
	\end{figure}
	\subsubsection{{Global Model Broadcasting}} 
	At the beginning of round $t$, the \ac{ps} broadcasts its aggregated model parameter $\br\itr{t}\in \setR^M$ that is quantized (or compressed) with $B$ bits per symbol. To this end, it encodes $\br\itr{t}$ into signal $\bxx\itr{t}\in \setR^M$ and sends it to the devices over the downlink broadcast channel in $M$ symbol intervals.
	For simplicity, we drop the round index $t$. 
	The \ac{ps} broadcasts $\bxx$ in $M$ intervals via linear beamforming, i.e., in time interval $m$ it transmits $x\dbc{m}\bww\dbc{m}$, where $\bww\dbc{m} \in \setC^N$ denotes its beamforming vector. Thus, device $k$ receives 
	$u_k \dbc{m} = \bh_k^\trp\dbc{m} \bww \dbc{m} x\dbc{m} + \xi_k\dbc{m}$, where $\bh_k\dbc{m}\in\setC^N$ denotes the channel coefficient between the \ac{ps} and device $k$ in interval $m$, and $\xi_k \dbc{m}$ is \ac{awgn} with mean zero and variance $\sigma_k^2$. 
	To address the restricted transmit power, we set  $\Ex{\abs{x\dbc{m}}^2}{} =1$ and
	$\tr{\bww \dbc{m} \bww^\her \dbc{m}} \leq P_{\mathrm{d}}$. The achievable spectral efficiency for device $k$ is given by
 \\
	\begin{align}
		R_k \dbc{m} = \log_2 \brc{1+ \frac{P_{\mathrm{d}}}{\sigma_k^2} \abs{\bh_k^\trp \dbc{m} \bww \dbc{m} }^2}, \label{eq:snr}
	\end{align}
 \\
	which needs to be more than $B$, i.e., minimum rate required for reliable downlink transmission.
	\subsubsection{{Local Model Updating}}  After receiving $\br\itr{t}$, each device updates its local model by running $\vartheta$ steps of stochastic gradient descent, i.e., device $k$ sets $\br_k^{\it{0}}=\br\itr{t}$, and updates $\br_k\itr{j+1}=\br_k\itr{j}-\zeta \nabla \mathcal{L}_k(\br_k\itr{j} \vert \mathcal{D}_k)$ for $j=0, 1, \cdots, \vartheta-1$ and some learning rate $\zeta$. \color{black} It then sets its local model to $\br_k=\br_k\itr{\vartheta}$ that is to be sent to the \ac{ps}. Without loss of generality, we consider that the devices address a supervised learning task. Each device possesses its own local training dataset, represented by $\mathcal{D}_k=\set{ \left(\mathbf{u}_{k,i}, v_{k,i}\right)}_{i=1}^{\vert \mathcal{D}_k \vert}$, where $\mathbf{u}_{k,i}$ is the $i$-th  data-point and $v_{k,i}$ is its corresponding label. The global dataset is the union of these local datasets, i.e., $\mathcal{D}=\cup_{k=1}^K \mathcal{D}_k$. The training minimizes the global loss function determined over $\mathcal{D}$ with respect to the model parameter, i.e., 
	\begin{align}
		\mathcal{L}(\br \vert \maD) \triangleq \sum _{k=1}^K \frac{|\mathcal{D}_k|}{|\mathcal{D}|} \mathcal{L}_k(\br \vert \mathcal{D}_k),
		\label{Loss}
	\end{align}
	where $\br\in \mathbb{R}^D$ is the $D$-dimensional model parameter to be learned, and $\mathcal{L}_k(\br \vert \mathcal{D}_k)$ represents the local loss function at device $k$ defined as $\mathcal{L}_k(\br\vert \mathcal{D}_k) \triangleq \frac{1}{|\mathcal{D}_k|} \sum_{i\in\maD_k } \ell_i \brc{\br\vert \mathcal{D}_k}$,
	with $\ell_i \brc{\br\vert \mathcal{D}_k}$ being the individual point-wise loss function.
 \begin{figure}
		\centering
		\vspace*{-.5cm}
		\hspace*{-.5cm}
		\includegraphics[scale=.35]{ 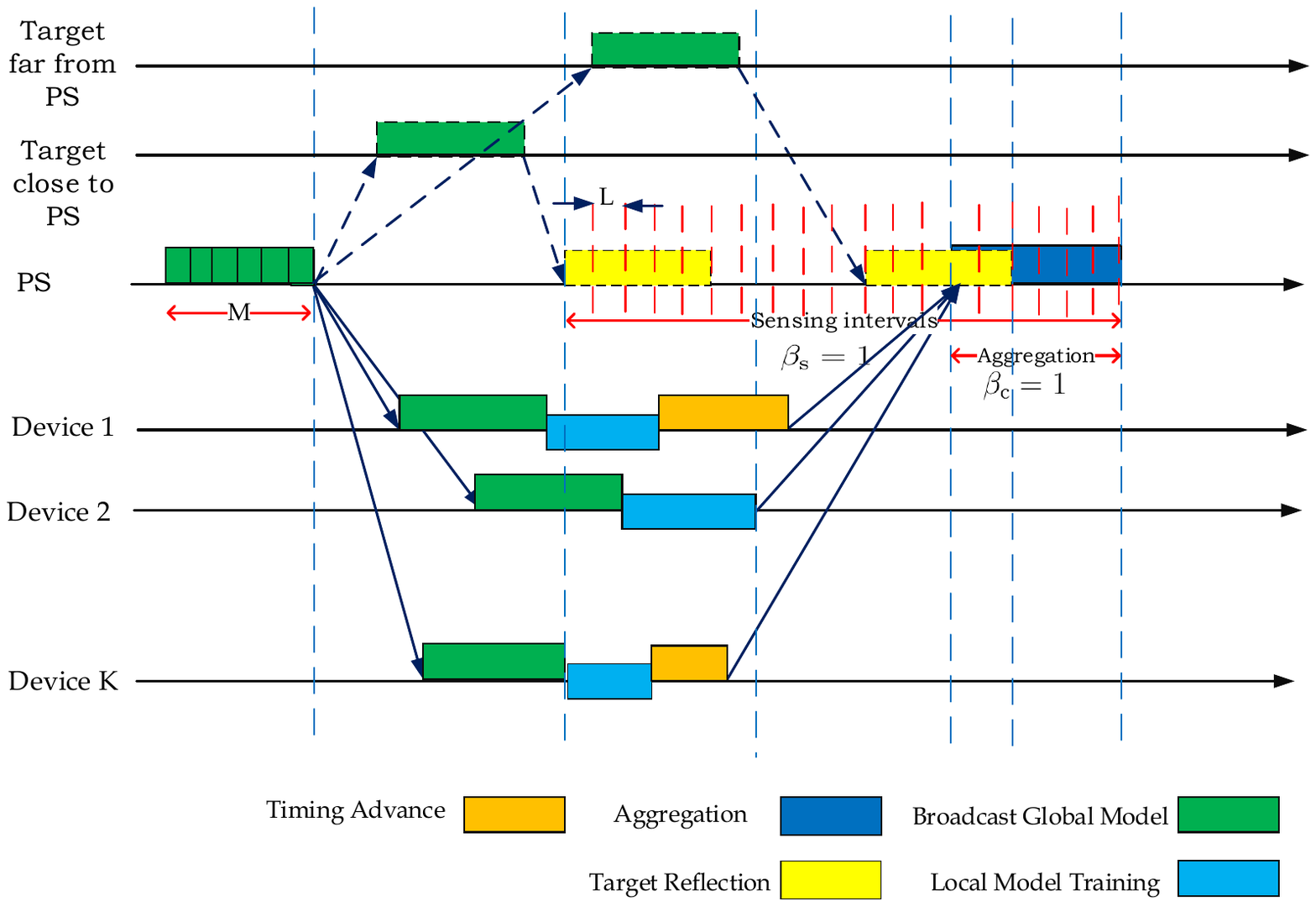}
		\vspace*{-.8cm}
		\caption{Schematic representation of communication in time domain.}
		\label{Fig_Schematic}
  \vspace*{-.5cm}
	\end{figure}
\subsubsection{Device Scheduling and Uplink Transmission}
The \ac{ps} schedules a subset of devices $\setS \subseteq \maK$ to transmit their model parameters. Leveraging the concept of OTA computation, the selected devices transmit their local model updates uncoded over the uplink channel in $M$ discrete intervals. To adhere to the transmit power constraints, each device scales its transmission accordingly. Specifically, in each interval $m$, the transmitted signal from device $k$ is given by:
\begin{align}
s_k[m] = b_k[m] r_k[m],
\end{align}
where $r_k[m]$ represents the local gradient from device $k$ in interval $m$ and is modelled as a zero-mean and unit variance random variable, and $b_k[m]$ is the scaling factor designed to ensure that the power does not exceed the uplink power limit $P_{\mathrm{u}}$, i.e., $\abs{b_{k}[m]}^2 \leq P_{\mathrm{u}}$. Since synchronization among distributed edge devices is essential in systems based on over-the-air computation, we employ the timing advance method commonly used in 4G Long-Term Evolution (LTE) and 5G New Radio technologies, as shown in Fig.~\ref{Fig_Schematic}. This method adjusts the transmission times of the devices so that their signals arrive simultaneously at the receiver, thereby minimizing the impact of propagation delays \cite{abari2015airshare, goldenbaum2013robust}.
	\subsubsection{ {OTA Model Aggregation with Integrated Sensing}} 
	We consider  a point-like target around the \ac{ps} whose echo  is received by the \ac{ps}.  As shown in Fig.~\ref{Fig_Schematic}, the target echo can be received in any communication round over either \textbf{(i)} an idle channel with no interference from the edge devices, or \textbf{(ii)} superposed along with the uplink transmission of the devices. 
We aim to exploit this overlap to
perform sensing and model aggregation simultaneously. Considering both reception cases, the signal received through the uplink channel in interval $m$ is written as follows: 
	\begin{align}
		\byy\dbc{m} &=  \beta_{\rm s} \dbc{m} 
		\mG 
		\bww \dbc{m} x \dbc{m} \nonumber\\&+  \beta_{\rm c} \dbc{m} \sum_{k\in \setS} \bff_k\dbc{m} b_k \dbc{m} \rmr_k \dbc{m} +\bzz \dbc{m},
		\label{eq:y}
	\end{align} 
	where $\bzz\dbc{m} \in \setC^{N}$ is the \ac{awgn} with mean zero and covariance matrix $\sigma_{\rm PS}^2\mI$. $\bff_k\dbc{m} \in \setC^{N}$ is the uplink channel coefficient from device $k$ to the \ac{ps} and $\mG \in \setC^{N\times N}$  is the unknown end-to-end  channel from the \ac{ps} to the target and backwards. 	
	$\beta_{\rm s}\dbc{m}$ is a binary variable denoting \textit{the activity of the echo signals} in a given time interval, i.e., $\beta_{\rm s}\dbc{m}=0$ when no echo is received by the \ac{ps} and $\beta_{\rm s}\dbc{m}=1$ when the \ac{ps} receives echoes from the target, and $\beta_{\rm c}\dbc{m}$ represents the activity of uplink transmissions in time interval $m$. 
	We consider that $\beta_{\rm s}\dbc{m}$ and $\beta_{\rm c}\dbc{m}$ are known apriori at the \ac{ps} since the \ac{ps} in practice can determine the minimum and maximum duration within which it receives echoes from targets \cite{martone2015passive, he2010mimo}. The \ac{ps} also knows the time at which it receives uplink transmissions, as it needs to be synchronized for OTA computation \cite{abari2015airshare, goldenbaum2013robust}.  

	\par Our ultimate goal is to develop a computationally-feasible algorithmic approach that extracts the aggregated \textit{global} model and the target unknowns directly from $\byy$.
    To this end, \ac{ps} invokes successive cancellation to extract sufficient statistics for sensing target parameters and aggregating model parameters over-the-air\color{black}. In each round~of~the FEEL setting, the \ac{ps} coordinates the iterative process. It collects model updates from the devices, aggregates them into the global model, and then  after $M$ intervals redistributes~the estimate aggregate model $\mathbf{\hat{r}}=[\hat{r}[1], \cdots, \hat{r}[M]]$ back to all~the devices. This iterative process continues until~the~global model converges.
 The minimizer of the global~loss function $\mathcal{L}(\br \vert \maD)$ is then approximated by this converging point. 
 \vspace{-0.3cm}
	\section{Cram\'er-Rao bound and Aggregation Error}
	In this section, we characterize the Cram\'er-Rao bound for target estimation, a lower bound on the variance of the unbiased estimator. Moreover, we utilize the \ac{mse} to calculate the aggregation error of the OTA-FEEL with integrated sensing.
	
	 As \eqref{eq:y} shows, the received signal at the \ac{ps} contains (1) \textit{local} model parameters, and (2) the unknown target response matrix $\mG$. The goal of the \ac{ps} is to extract these parameters.
	From the information-theoretic perspective, it is efficient to infer all the unknowns jointly from $\byy$. This is however computationally complex. We hence follow the  successive approach \cite{cover1991network}:
	given $L$ observations of the received signal, the \ac{ps} first performs sensing, i.e., it estimates the response matrix $\mG$ as $\hat{\mG}$ using \ac{ml} estimation. It then cancels out the sensing interference and uses the interference-reduced signal $\bar{\byy}\dbc{\ell} = {\byy}\dbc{\ell} - \hat{\mG} \bww \dbc{\ell} x \dbc{\ell}$ 	for $\ell \in \dbc{L}$ and conducts post-processing to estimate the global model. With this successive approach, the error of the estimated global model is proportional to the error of target estimation, i.e., it is proportional to the \ac{mse} term $\Ex{\norm{\mG - \hat{\mG}}^2 }{ }$: the higher this error is, the more interference remains after cancelling the sensing signal. We hence consider this error as the metric to quantify the quality of sensing. 
	It is noteworthy that, in general, the order of the model aggregation process and sensing can be changed. However, it is  natural to perform the sensing before estimating the aggregated model. This follows from the fact that the \ac{ps}'s main task is to estimate the aggregated model. In the following subsections, we first estimate the target response matrix and then characterize the Cram\'er-Rao bound and aggregation error.
\vspace{-0.3cm}
\subsection{Target Response Matrix Estimation}
Given the block of $L$ observations $\byy\dbc{1}, \ldots, \byy\dbc{L}$, we intend to estimate the unknown response matrix $\textbf{G}$.
The received signal at the PS in interval $\ell$ can be compactly written as:
\begin{align}
	\byy\dbc{\ell} &=  
	\mG 
	\bww \dbc{\ell} x \dbc{\ell} +  \tilde{\bzz} \dbc{\ell},
\end{align} 
where $\tilde{\bzz}\dbc{\ell}\in \setC^{N}$ describes the effective noise including the uplink interference of the devices, i.e., $\tilde{\bzz} \dbc{\ell}= 
	\beta_{\rm c} \dbc{\ell} \sum_{k\in \setS} \bff_k\dbc{m} b_k \dbc{m} \rmr_k \dbc{\ell} + \bzz \dbc{\ell}$, which can be compactly written as 
\begin{align}
	\tilde{\bzz} \dbc{\ell}=
	\beta_{\rm c} \dbc{\ell} \mF \mB_\setS  \br \dbc{\ell} + \bzz \dbc{\ell},
\end{align}

with $\mF = \dbc{\bff_1, \ldots, \bff_K}\setC^{N\times K}$ being the uplink channel matrix, $\mB_\setS$ being the scaling matrix of the devices, i.e., $\mB_\setS = \Diag{b_1, \ldots, b_K}$ with $b_k= 0$ for $k\notin \setS$, and $\br\dbc{\ell}$ being the vector of local models, i.e., $\br\dbc{\ell} = \dbc{\rmr_1\dbc{\ell}, \ldots, \rmr_K\dbc{\ell}}^\trp$. Considering the  Gaussian model for the communicated local parameters in the uplink channels \footnote{Note that though this Gaussian assumption is practically accurate (see \cite{lee2020bayesian}), one can further consider it as the worst-case assumption following the maximum entropy property of Gaussian distribution \cite{Cover}.}, the effective noise is an additive Gaussian noise process, as described in Lemma~\ref{lem:1}.

\begin{lemma}[Distribution of $\tilde{\bzz} \dbc{\ell}$]
	\label{lem:1}
	The effective sensing noise process $\tilde{\bzz} \dbc{\ell}$ is zero-mean Gaussian with covariance matrix
	\begin{align}
		\mR= \beta_{\rm c}  \dbc{\ell} \mF \mB_\setS \mB_\setS ^\her \mF^\her  + \sigma_{\rm PS}^2 \mI. \label{eq:R}
	\end{align} 
\end{lemma}
\begin{proof}
	The proof is given in \textbf{Appendix A}.
\end{proof}
We now derive sufficient statistics from $\byy\dbc{\ell}$: since $\mR $ is a positive-definite matrix, we can whiten the noise process~via the standard whitening filter. Let $\mR = \mU \boldsymbol{\Lambda} \mU^\her $ be the eigenvalue decomposition of $\mR$. The \ac{ps} whitens the noise using the linear filter
$\mT = \mU  \boldsymbol{\Lambda}^{-1/2}  \mU^\her$. This means that it determines the sufficient statistics 
$\tilde{\byy}\dbc{\ell} = \mT \byy \dbc{\ell}$ from the received signals for $\ell \in\dbc{L}$. The sufficient statistics are hence given by
\begin{align}
	\tilde{\byy}\dbc{\ell} = \mT \mG 
	\bww \dbc{\ell} x \dbc{\ell} + \breve{\bzz} \dbc{\ell},\label{eq:Ysens2}
\end{align}
where $\breve{\bzz} \dbc{\ell} \sim \mathcal{CN}\brc{ \boldsymbol{0}, \mI }$. 
We now consider the deterministic unknown matrix $\mG$ and write its likelihood function as
\begin{align}
	\mal\brc{\mG} = \prod_{\ell=1}^{L} p\brc{\tilde{\byy}\dbc{\ell} \vert \mG}, \label{eq:loglik}
\end{align}
where $p\brc{\tilde{\byy}\dbc{\ell} \vert \mG} = \pi^{-N} \exp\set{-\norm{\tilde{\byy}\dbc{\ell}  - \mT\mG \bww \dbc{\ell} x \dbc{\ell}}^2 }
$. The log-likelihood function is then given by
\begin{align}
	\log\mal\brc{\mG} &= -LN\log \pi - \sum_{\ell=1}^{L} \norm{\tilde{\byy}\dbc{\ell}  - \mT\mG 
		\bww \dbc{\ell} x \dbc{\ell}},\nonumber\\&= -LN\log \pi -\norm{\tilde{\mathbf{Y}}-\mT\mG \mathbf{W}\mathbf{D}}^2_F,
	\label{eq:logliklihood}
\end{align}
where $\tilde{\mathbf{Y}}=[\tilde{\byy}\dbc{1}, \cdots, \tilde{\byy}\dbc{L}]$, $\mathbf{W}=[\bww \dbc{1}, \cdots, \bww \dbc{L}]$ and $\mathbf{D}=\Diag{x \dbc{1}, \cdots, x \dbc{L}}$. The \ac{ml} estimator for $\mathbf{G}$ is given by $\hat{\mathbf{G}}= \argmax_\mathbf{G} \log\mal\brc{\mG} $ which is 
	$\hat{\mathbf{G}}=\mT^{-1} \tilde{\mathbf{Y}} \brc{\mathbf{W} \mathbf{D}}^\dagger$, since $\mathbf{W} \mathbf{D}$ has pseudo-inverse.
\vspace{-0.3cm}
\subsection{Deriving the Cram\'er-Rao bound}
 In this section, we assess the sensing quality using the Cram\'er-Rao bound \cite{song2024cramer, hua2023mimo, hua2022mimo}.
In a nutshell, the Cram\'er-Rao theorem indicates that when  using an optimal unbiased estimator, the error $\mG - \hat{\mG}$ converges to a Gaussian process as $L$ grows large. The mean of this error is zero, since the estimator is unbiased, and its covariance matrix is given by the inverse of the \textit{Fisher information matrix}, which measures the amount of information that the data provides regarding the parameter of interest. 
\color{black}
Given the parametric unknown $\mG$, the {Fisher information matrix} is defined as
$
I\brc{\mG} = - \Ex{\nabla^2 \log \mal\brc{\mG} }{},
$
where $\mal\brc{\mG}$ is the likelihood function defined in \eqref{eq:loglik} \cite{kay1993fundamentals}.
Here, $\nabla^2$ is the second-order gradient with respect to $\mG$, and expectation is taken with respect to all random variables. Note that the Fisher information matrix $I\brc{\mG} $ is an $N^2 \times N^2$ matrix containing all second-order derivatives of the likelihood function. 
We specify the Fisher information matrix and the Cram\'er-Rao bound in Proposition~\ref{propos:1}.
\begin{proposition}[Cram\'er-Rao bound] 
	\label{propos:1}
	The Fisher information matrix $I\brc{\mG}$ is independent of $\mG$ given by 
	$
	I\brc{\mG} =  2\sum_{\ell=1}^{L}  \mR^{-1} \otimes \bww\dbc{\ell} \bww^\her\dbc{\ell}.
	$
	This implies that for any unbiased estimator $\hat{\mG}$ of $\mG$, the error covariance is bounded via 
	\begin{align}
		\mC^\star = \frac{1}{2}\mR\otimes \brc{\mW^* \mW^\trp }^{-1},
	\end{align}
	where $\mW = \dbc{\bww\dbc{1}, \ldots, \bww\dbc{L}}$ denotes the downlink precoding matrix of the block.
	\end{proposition}
	\begin{proof}
		See \textbf{Appendix B}.
	\end{proof}
	Proposition~\ref{propos:1} implies the following result: let $\hat{\mG}$ be an unbiased estimator of $\mG$, and let~$\mC_{\hat{\mG}}$ denote the error covariance matrix, i.e., $\mC_{\hat{\mG}} = \Ex{ (\bvv -\hat{\bvv})  (\bvv-\hat{\bvv}) ^\her }{}$ with $\bvv = \mathrm{vec} \brc{\mG}^\trp$ and $\hat{\bvv} = \mathrm{vec} (\hat{\mG})^\trp$.  The error covariance matrix satisfies
	$
	\mC_{\hat{\mG}} \succeq \frac{1}{2}\mR\otimes \brc{\mW^* \mW^\trp }^{-1}. 
	$
	Using this bound, the sensing error $\Ex{\norm{\mG - \hat{\mG}}^2 }{ }$ can be bounded as $\Ex{\norm{\mG - \hat{\mG}}^2 }{ } \geq \mathrm{CRB}\brc{\setS,\mW}$, with $\mathrm{CRB}\brc{\setS,\mW}$ being the lower bound.
 	\begin{subequations}
	\begin{align}
		\mathrm{CRB}\brc{\setS,\mW} &= \tr{\mC^\star}\
		=  \frac{1}{2} \tr{\mR \otimes \brc{\mW^* \mW^\trp }^{-1}},\\&= \frac{1}{2} \tr{\mR}  \tr{\brc{\mW^* \mW^\trp }^{-1}}.
	\end{align}
 	\end{subequations}
 Here, $\setS$ and $\mW$ are parameters indicating the dependency on the set of active devices and the precoding matrix, respectively. We define $ \mathrm{CRB}\brc{\setS,\mW}$ as the metric for sensing quality, where smaller values indicate higher sensing accuracy. However, as the size of $\setS$ increases, $\mathrm{CRB}\brc{\setS ,\mW}$ also increases. This trade-off implies that achieving higher learning accuracy comes at the expense of degraded sensing. Hence, to maintain a certain level of sensing quality, optimization of the number of active devices participating in OTA-FEEL is necessary.
	
\vspace{-0.4cm}
	\subsection{Aggregation Error Analysis}
 During aggregation, the received signal at the PS can be given by setting $\beta_{\rm c}\dbc{\ell} = 1$ as  
	$
	\byy\dbc{\ell} =  
	\beta_{\rm s} \dbc{\ell} \mG 
	\bww \dbc{\ell} x \dbc{\ell} + \mF \mB_\setS  \br \dbc{\ell} + \bzz \dbc{\ell}.
	$
	To aggregate the local models, the \ac{ps} first uses the estimated $\mG$ to cancel out the estimated sensing signal, i.e., it determines
	\begin{subequations}
		\begin{align}
			\bar{\byy}\dbc{\ell} &= {\byy}\dbc{\ell} - \hat{\mG} \bww \dbc{\ell} x \dbc{\ell},\\
			&= \beta_{\rm s} \dbc{\ell} \brc{\mG - \hat{\mG}}
			\bww \dbc{\ell} x \dbc{\ell} +  \mF \mB_\setS  \br \dbc{\ell}  + \bzz \dbc{\ell}.
		\end{align}
	\end{subequations}
It then invokes the concept of analog function computation and leverages the linear superposition of the multiple access channel to estimate the global model directly from the received signal via a linear receiver. \color{black} Let $\bc$  denote the receive beamforming vector such that  $\norm{\bc} = 1$. Denoting $\eta$ as the power scaling, the aggregated model parameter ${\rmr}[\ell]$ in interval $\ell$ is estimated as: 
		\begin{align}
			\hat{\rmr} &\dbc{\ell}	= \frac{\bc^\her \bar{\byy}\dbc{\ell}  }{\sqrt{\eta}}
			=  \beta_{\rm s} \dbc{\ell}  \frac{\bc^\her  }{\sqrt{\eta}} \brc{\mG - \hat{\mG}}
			\bww \dbc{\ell} x \dbc{\ell}\nonumber \\&+  \frac{\bc^\her  }{\sqrt{\eta}}  \mF \mB_\setS  \br \dbc{\ell}   +  \frac{\bc^\her  }{\sqrt{\eta}}\bzz \dbc{\ell}=\frac{\bc^\her  }{\sqrt{\eta}} \mF \mB_\setS  \br \dbc{\ell}   + \rme \dbc{\ell},
		\end{align} 
	where $\rme \dbc{\ell} $ denotes the stochastic aggregation error process and is defined as
	\begin{align}
		\rme \dbc{\ell} =  \frac{ \beta_{\rm s} \dbc{\ell} \bc^\her  }{\sqrt{\eta}} \brc{\mG - \hat{\mG}}
		\bww \dbc{\ell} x \dbc{\ell} +  \frac{\bc^\her  }{\sqrt{\eta}}  \bzz \dbc{\ell}. 
	\end{align}
	The asymptotic of this error are described in Lemma~2.
	\begin{lemma}
		\label{Lemma2}
		As $L$ grows large, the error term $\rme \dbc{\ell} $ tends to converge to a zero-mean Gaussian random variable with~variance
		\begin{align}
			\Ex{\abs{\rme \dbc{\ell}}^2}{} =\dfrac{ \beta_{\rm s} \dbc{\ell}}{\eta} \brc{\bc^\her\otimes \bww^\trp \dbc{\ell}} \mC \brc{\bc \otimes \bww^* \dbc{\ell}}
			+  \dfrac{ \norm{\bc}^2 }{\eta}\sigma_{\rm PS}^2. \nonumber
		\end{align}
	\end{lemma}
	\begin{proof}
		See \textbf{Appendix C}.
	\end{proof}
The aggregation error is characterized in terms of the error process statistics. We assess the aggregation error using \ac{mse}. The average aggregation error over the block is given by
	\begin{subequations}
		\begin{align}
			&\mae \brc{\bc, \mB_\setS}
			= \frac{1}{L} \sum_{\ell=1}^{L}  \Ex{ \abs{\hat{\rmr} \dbc{\ell} - \rmr \dbc{\ell}}^2 }{ },
			\\&= \frac{1}{L} \sum_{\ell=1}^{L}  \Ex{ \abs{\dfrac{\bc^\her}{\sqrt{\eta}} \bar{\byy}\dbc{\ell}-\sum_{k\in\setS} \phi_k \rmr_k \dbc{\ell}}^2 }{ },\\
			&= \frac{1}{L} \sum_{\ell=1}^{L} \Vert \dfrac{1}{\sqrt{\eta}} \bc^\her \mF \mB_\setS - \boldsymbol{\phi}_\setS^\trp \Vert^2+ \underbrace{\frac{1}{L} \sum_{\ell=1}^{L}\Ex{ \vert \rme \dbc{\ell}\vert ^2}{}}_{E_2},\label{eq:Epps_last}
		\end{align}
	\end{subequations}
	where $\rmr \dbc{\ell} = \sum_{k\in\setS} \phi_k \rmr_k \dbc{\ell}$ is the desired aggregate model and $\bphi_\setS \in \setR^K$ is the vector of aggregation weights being $\phi_k$ for $k \in \setS$ and zero elsewhere. Invoking Lemma \ref{Lemma2}, we have 
		\begin{align}
   E_2 = \frac{1}{\eta L} \sum_{\ell=1}^{L} \beta_{\rm s}\dbc{\ell}   \brc{\bc^\her \otimes \bww^\trp \dbc{\ell}} \mC \brc{\bc \otimes \bww^* \dbc{\ell}}
			+   \sigma_{\rm PS}^2 \dfrac{ \norm{\bc}^2 }{\eta},\nonumber
   \end{align}
   \begin{align}
			&= \frac{\beta_{\rm s}\dbc{\ell}}{\eta L}  \tr{ \mC  \sum_{\ell=1}^{L} \brc{\bc \otimes \bww^* \dbc{\ell}} \brc{\bc^\her\otimes \bww^\trp \dbc{\ell}}}
			+  \sigma_{\rm PS}^2 \dfrac{\norm{\bc}^2 }{\eta}. \nonumber 
		\end{align}
	By using identity $\brc{\mA \otimes \mB} \brc{\mC \otimes \mD} = {\mA\mC \otimes \mB\mD}$, we have
		\begin{align}
			\sum_{\ell=1}^{L} \brc{\bc \otimes \bww^*\dbc{\ell}} \brc{\bc^\her \otimes \bww^\trp \dbc{\ell}} &=   \bc\bc^\her \otimes \brc{\sum_{\ell=1}^{L}\bww^* \dbc{\ell}\bww^\trp \dbc{\ell}},\nonumber
			\\&=   \bc\bc^\her \otimes \mW^* \mW^\trp. \label{eq:tempError}
		\end{align}
	By substituting \eqref{eq:tempError} in $E_2$, and after some lines of basic derivations, we~have
	\begin{align}
		E_2 &= \frac{\beta_{\rm s}\dbc{\ell} }{2 \eta L}  \tr{ \mR \bc\bc^\her \otimes \mI_N}
		+ \dfrac{\sigma_{\rm PS}^2}{\eta} \norm{\bc}^2. \label{eq:eqtemp}
	\end{align}
	Using the property $\tr{\mA\otimes\mB} = \tr{\mA} \tr{\mB}$, \eqref{eq:eqtemp} reduces~to
		\begin{align}
			E_2 
			=  \bc^\her \brc{ \frac{ N \beta_{\rm s}\dbc{\ell} }{2 \eta L}   \mR  + \dfrac{\sigma_{\rm PS}^2 }{\eta}\mI} \bc. \label{eq:eqtemp2}
		\end{align}
	Plugging \eqref{eq:eqtemp2} into  \eqref{eq:Epps_last}, the aggregation error reads 
	\begin{align}
		\mae \brc{\bc, \mB_\setS}&= \bc^\her \brc{ \frac{ N \beta_{\rm s} \dbc{\ell} }{2 \eta L}   \mR + \dfrac{\sigma_{\rm PS}^2}{\eta} \mI} \bc\nonumber\\& +\frac{1}{L} \sum_{\ell=1}^{L}\Vert \dfrac{1}{\sqrt{\eta}} \bc^\her \mF \mB_\setS - \boldsymbol{\phi}_\setS^\trp \Vert^2.
  \label{eq:pop}
	\end{align}
	
	According to \eqref{eq:pop}, the aggregation error depends on the choice of scaling factors, i.e., $b_k$ for $k\in\setS$, $\eta$, the subset $\setS$, and post-processing unit $\bc$. We choose these transmit and receive parameters, i.e., $\mB$ and $\eta$, using the zero-forcing coordination strategy \cite{yang2020federated}. 
	In zero-forcing coordination, the active devices set their scaling factors to
		$b_k = \frac{ \phi_k\sqrt{\eta}} {\bc^\her \bff_k }$, and the \ac{ps} sets $\eta$ to
	$\eta = P_{\rm u} \min_{k\in\setS} {\abs{\bc^\her \bff_k}^2 }/{ \phi_k^2} $ to satisfy the transmit power constraint. This way, they zero-force the residual term, i.e., the first term in \eqref{eq:Epps_last}, and hence the $\rme\dbc{\ell}$ is the only error term. With the expression of $\mR$ from \eqref{eq:R}, we can show that for these choices of $\eta$ and $b_k$, the aggregation error is given~by
	\begin{align}
		\mae\brc{\bc, \setS} = \frac{1+N/L}{P_{\rm u} /\sigma^2 }\max_{k\in\setS} \frac{ \phi_k^2 \norm{\bc}^2 }{\abs{\bc^\her \bff_k}^2 } + \frac{N}{ L} \sum_{k\in\setS} \phi_k^2.
	\end{align} 
	We use $\mae\brc{\bc, \setS}$ to represent the quality of OTA aggregation.
    \vspace{-.3cm}
	\section{Device Scheduling in OTA-FEEL with Integrated Sensing}
	In this section, we formulate the device scheduling problem in OTA-FEEL with integrated sensing. A larger number of devices typically contribute positively to the model updates in OTA-FEEL. Thus, there is potential to improve the convergence rate significantly by incorporating more users with diverse data. However, a higher user count introduces challenges, such as increased aggregation error during model updates. Furthermore, the increased interference on the sensing signal caused by a larger number of devices can result in a poorer estimation of the unknowns related to the targets. This interference negatively impacts the quality of sensing data and, consequently, the accuracy of target estimation. Hence, optimizing OTA-FEEL with integrated sensing involves careful consideration of the sensing-learning trade-off, seeking to maximize participation while keeping aggregation error and sensing error below a predefined threshold. Using the developed metrics for aggregation and sensing error, the above constrained optimization can be mathematically written as optimization $\mathcal{P}_1$ given at the top of the next page. 
 \begin{figure*}[t]
	\begin{subequations}
		\begin{align}	(\mathcal{P}_1): \max_{\mW\in\setC^{N\times L},  \setS, \bc \in\setC^N}   \abs{\setS} \qquad
			\text{subject to } \qquad & (\mathrm{C}_1)~\mathrm{CRB} \leq \Gamma_0,
			 (\mathrm{C}_2)~\mae \brc{\bc, \mB_\setS} \leq \epsilon_0,
			 (\mathrm{C}_3)~\norm{\bc}_2 = 1,\\
			& (\mathrm{C}_4)~\frac{1}{L} \tr{ \mW^* \mW^\trp } \leq P_{\mathrm{d}}, \label{C3} 
			(\mathrm{C}_5)~
   \frac{1}{K} \sum_{k=1}^K \frac{\bh_k^\trp\mW\mW^\her \bh_k^*}{L} \geq \gamma,
		\end{align}	\label{eq:opt}
	\end{subequations}
 \hrule
  \end{figure*}
In this problem, $(\mathrm{C}_1)$ and $(\mathrm{C}_2)$ constrain the sensing quality and aggregation fidelity, respectively, where  $\epsilon_0$ and $\Gamma_0$ represent the maximum tolerable aggregation error and target estimation error, respectively. The constraint $(\mathrm{C}_3)$  avoids getting $\bc = \boldsymbol{0}$, which introduces singularity. The constraint $(\mathrm{C}_4)$ limits the average downlink transmit power over a transmission  block of $L$ symbols.  Finally, the constraint $(\mathrm{C}_5)$ ensures the minimum downlink \ac{snr} scaling, which is directly proportional to the \ac{snr} itself. This requirement guarantees that all devices receive the updated global model without any errors. In other words, the \ac{ps} must ensure that the average downlink SNR, as defined in \eqref{eq:snr}, over a transmission block of $L$ symbols exceeds a specified threshold. This can be represented as
	$
		\min_k \frac{1}{L} \sum_{\ell} \abs{\bh_k^\trp\dbc{\ell} \bww \dbc{\ell} }^2 \geq \gamma,
	$
	where $\gamma$ is the minimum required \ac{snr} scaling. Considering the fact that $L$ is smaller than the coherence time interval of the channel, we can drop the time index of the channel coefficients, i.e., $\ell$ in $\bh_k^\trp\dbc{\ell}$. This leads to the constraint given in $(\mathrm{C}_5)$.
Note that according to $(\mathcal{P}_1)$, the downlink beamforming matrix $\mathbf{W}$ influences the scheduling through $(\mathrm{C}_1)$, $(\mathrm{C}_4)$ and $(\mathrm{C}_5)$. 
\par Hence, this integrated design improves system performance by optimizing beamforming for both learning and sensing needs, thereby enhancing sensing and learning performance while managing uplink interference.
The optimization problem $\mathcal{P}_1$ is a mixed-integer  non-linear programming (MINLP), which is an NP-hard problem. This problem is highly intractable  due to its combinatorial objective function, i.e., $\vert \setS\vert$ and the non-convex constraints $(\mathrm{C}_1)$ and $(\mathrm{C}_2)$ with coupled combinatorial  $\setS$ and continuous variables $\bc$ and $\mW$. Hence, to achieve efficient computation,  in the following, we use  hierarchical optimization  where the outer layer tackles the device scheduling task, and the inner layer  optimizes the transmit and receive beamforming at the PS. 
    \vspace{-.3cm}
		\section{Hierarchical Optimization via  Matching Pursuit}
In this section, we employ techniques from the literature on compressive sensing and sparse recovery algorithms to address the problem. Matching pursuit, a prevalent category of sparse recovery algorithms, iteratively reconstructs the signal by selecting the most error-reducing component in each iteration \cite{mallat1993matching}. The algorithm starts with an empty set of signal entries and adds new indices in each iteration based on the correlation with the residual signal. This approach maximizes the reduction of the objective function compared to the previous iteration. Matching pursuit, belonging to step-wise regression methods, can be customized for our problem.
\vspace{-.5cm}
 \subsection{Matching Pursuit-based Hierarchical Optimization}
	\label{sec:nutshell}
 We formulate  the device scheduling problem $\mathcal{P}_1$  as a sparse recovery problem. Our ultimate objective is to derive the sparsest set of devices that when deactivated, ensure that the minimum computation error and Cram\'er-Rao bound remain below an acceptable level and satisfy the power constraints represented through $(\mathrm{C}_4)$ and $(\mathrm{C}_5)$.
 Our iterative algorithm follows the step-wise regression technique, i.e., \textbf{(i)} We initiate our search by selecting the maximal set of devices, i.e., $\setS = \dbc{K}$. \textbf{(ii)} For the chosen set of devices, we solve both marginal problems to find the optimal transmit and receive beamforming design, i.e., $\mW$ and $\bc$, respectively. The marginal objectives are proportional to the constraints in the original problem. \textbf{(iii)} Notably, for a given $\setS$, the design problems over $\mW$ and $\bc$ become independent. After solving the marginal problems, we examine the feasibility of the solution, i.e., we check if the derived $\mW$ and $\bc$ satisfy the constraints in $\maP_1$. If infeasible, we remove the device whose removal maximally reduces our \textit{selection metric}. The algorithm stops in the case of feasibility. If the constraints are not satisfied, we exclude the \textit{least effective} device and repeat the procedure. Let the index of this device be $k^\star$. After updating $\setS \leftarrow \setS - \set{k^\star}$, we repeat the procedure until we reach a feasible solution.
	
	The key design points in this algorithm are the specification of the marginal problems and the selection metric. We address these two design tasks in the sequel.
    \vspace{-.4cm}
	\subsection{Marginal Problem for Downlink Precoding}
	Considering $\setS$ to be fixed, the precoder $\mW$ is to be chosen such that $(\mathrm{C}_1)$, $(\mathrm{C}_4)$ and $(\mathrm{C}_5)$ are satisfied. These constraints describe a feasible region that can generally be empty. Noting that the Cram\'er-Rao bound is an increasing function of $\abs{\mas}$,\footnote{The proof is omitted for brevity.}  the emptiness of this region suggests the removal of a client from $\setS$. We use this fact later to derive the selection metric.

 Next, we derive the marginal problem by selecting a precoder from the feasible set for a nonempty feasible region. we address this task through the following lemma.
	
	\begin{lemma}
		\label{lem3}
		Let the set $\setW$ be described by the constraints $f_i\brc{\mW} \leq F_i$ for $I$ different functions $f_i\brc\cdot$ with $i\in \dbc{I}$. Let $\mW_0$ be the solution to the following optimization problem
		\begin{align}
			\min_\mW f_j\brc{\mW} \qquad \text{subject to} \ f_i\brc{\mW} \leq F_i \; \; \text{for} \; i \neq j,\label{optim:lem}
		\end{align}
		for some $j \in \dbc{I}$. Then, $\setW$ is nonempty, iff $f_j\brc{\mW_0} \leq F_j$.
	\end{lemma}
	\begin{proof}
		The proof of the forward argument follows directly from the definition of $\setW_0$: let $\mW_0$ be the solution to \eqref{optim:lem}. Hence, we have $f_i\brc{\mW_0} \leq F_i$ for $i \neq j$. Given $f_j\brc{\mW_0} \leq F_j$, we can conclude that $\mW_0\in \setW$, and hence $\setW \neq \emptyset$.
		
		For proof of converse, we start with defining the set $\setW_j$ to be the feasible set of \eqref{optim:lem}, i.e.,
		$
			\setW_j = \set{ \mW:  f_i\brc{\mW} \leq F_i \; \; \text{for} \; i \neq j }.
	$
		For $\mW \in \setW_j$, we have $f_j \brc{\mW} \geq f_j \brc{\mW_0}$. Note that $\setW \subseteq \setW_j$, we conclude $f_j \brc{\mW} \geq f_j \brc{\mW_0}$ also for $\mW \in \setW$. On the other hand, for $\mW \in \setW$, we have, by definition, $f_j \brc{\mW} \leq F_j $.  This implies that for $\mW \in \setW$,
	$		f_j \brc{\mW_0} \leq 	f_j \brc{\mW} \leq F_j. 
	$
		For $\setW \neq \emptyset$, this concludes that $f_j\brc{\mW_0} \leq F_j$. 
	\end{proof}
	
	We can utilize the above lemma as follows: for a given $\setS$, there exists a feasible precoder $\mW$ that satisfies the constraints $(\mathrm{C}_1)$, $(\mathrm{C}_4)$ and $(\mathrm{C}_5)$ if and only if $\mathrm{CRB} \brc{\setS, \mW^\star} \leq \Gamma_0$, where $\mW^\star$ is the solution to the following optimization problem:
 \begin{subequations}
	\begin{align}
		\brc{\mam_1}~&\min_{\mW}  \mathrm{CRB} \brc{\setS, \mW}\nonumber\\ 
		\text{s. t. } 
		&\frac{1}{L} \tr{ \mW \mW^\her } \leq P_{\rm d},
		~\gamma  -  \frac{1}{K} \sum_{k=1}^K \frac{\bh_k^\trp\mW\mW^\her \bh_k^*}{L}  \leq 0.\nonumber
	\end{align}
	\end{subequations}
	We hence consider $\mam_1$ as the first marginal problem that minimizes the Cram\'er-Rao bound subject to power and quality constraints. Its solution determines whether further reduction in the number of clients is necessary. $\mam_1$ minimizes the achievable sensing error, while maintaining the power budget and ensuring a minimum tolerable \ac{snr} scaling averaged over time and clients. We solve the marginal problem using \ac{lmi} programming., i.e., we define the variable $\mQ = \mW\mW^\her$ to write $\mam_1$ as:
\begin{subequations}
		\begin{align}
		&\min_{\mQ}  \frac{1}{2} \tr{\mR}  \tr{\mQ^{-1} } \nonumber \\
		&\text{s.t. }  
		~\frac{1}{L} \tr{ \mQ } \leq P_{\rm d}, 
		\gamma  -  \frac{1}{K} \sum_{k=1}^K \frac{\bh_k^\trp\mQ \bh_k^*}{L}  \leq 0, ~\mQ \succeq 0.\nonumber
	\end{align}
\end{subequations}
As ${\mR}$ does not depend on $\mQ$, we can drop $\tr{\mR}$. By defining the auxiliary variable $\mX \succeq \mQ^{-1}$, we write the above optimization as \cite{boyd2004convex}:
\begin{subequations}
		\begin{align}
		&\min_{\mQ}   \tr{\mX } \nonumber\\
		&\text{s.t. }
		~\frac{1}{L} \tr{ \mQ } \leq P_{\rm d}, \quad
		\gamma  -  \frac{1}{K} \sum_{k=1}^K \frac{\bh_k^\trp\mQ \bh_k^*}{L}  \leq 0,\nonumber \\ &\qquad \mQ \succeq 0 \; ~\text{and}~ \; \mX \succeq \mQ^{-1}\nonumber.
	\end{align}
\end{subequations}
	Using Schur's complement argument, the problem reduces to  \ac{lmi} problem:
\begin{subequations}
		\begin{align}
		&\min_{\mX, \mQ} \qquad \tr{ \mX }\nonumber \\
		&\text{s.t. } \qquad 
		\frac{1}{L} \tr{ \mW \mW^\her } \leq P_{\rm d}, \\
		&\gamma  -  \frac{1}{K} \sum_{k=1}^K \frac{\bh_k^\trp\mW\mW^\her \bh_k^*}{L}  \leq 0, \\
		&\mQ \succeq 0 \; \text{ and }\; 
		\begin{bmatrix}
			\mX &\mI\\
			\mI &\mQ
		\end{bmatrix} \succeq 0.
	\end{align}
\label{eq:W_star}
\end{subequations}
	This latter form can be readily solved using standard semi-definite programming (SDP), e.g., using \texttt{cvx} toolbox in \texttt{MATLAB} \cite{grant2014cvx}. Interestingly, this problem does not depend on $\setS$ and $\bc$, and its solution remains unchanged as the algorithm iterates over $\setS$. In fact, the dependency of the Cram\'er-Rao bound on $\setS$ is only through $\mR$, which does not appear in this marginal problem. This indicates that the design of precoder decouples from the design of the post-processing unit. This leads to a significant reduction in complexity.

        \vspace{-.3cm}
	\subsection{Marginal Problem for Receive Beamforming Design}
	For a given $\setS$, the design of the receiver reduces to finding a vector $\bc$ that satisfies the constraints $(\mathrm{C}_2)$ and $(\mathrm{C}_3)$. Starting with constraint $(\mathrm{C}_2)$, we are looking for a vector $\bc$ that fulfills
	\begin{align}
		\mae\brc{\bc,\setS} = \frac{1+N/L}{P_{\rm u} /\sigma^2 }\max_{k\in\setS} \frac{ \phi_k^2 \norm{\bc}^2 }{\abs{\bc^\her \bff_k}^2 } + \frac{N}{ L} \sum_{k\in\setS} \phi_k^2 \leq \epsilon_0.
	\end{align} 
	The above constraint is equivalently represented by the  set of constraints $C_k\brc{\bc}	\leq 0$	for $k \in \setS$, where $C_k\brc{\bc}$ is defined as
	\begin{align}
		C_k\brc{\bc} =
		{ \phi_k^2 \norm{\bc}^2 }
		+  \frac{P_{\rm u}  }{\sigma^2\brc{1+N/L}} \brc{\frac{N}{ L} \sum_{k\in\setS} \phi_k^2 - \epsilon_0 } {\abs{\bc^\her \bff_k}^2 }.
		\label{eq:C_k}
	\end{align}
Following \cite{bereyhi2023device}, we combine the individual constraints into a single linear constraint using the following lemma.
	\begin{lemma}
		\label{lem4}
		Define ${C^\star}\brc{\bc}$ as
		\begin{align}
			C^\star \brc{\bc} = \max_{\tau_k\in\dbc{0,1}: k\in\setS}\sum_{k\in\setS} \tau_k C_k\brc{\bc}.
		\end{align}
		The set of constraints $C_k\brc{\bc}	\leq 0 $ for $k \in \setS$ at point $\bc$ is satisfied if and only if ${C^\star}\brc{\bc}  \leq 0$.
	\end{lemma}
	\begin{proof}
	The proof is straightforward: assume that all the~constraints are satisfied at $\bc$. Since $\tau_k$ is non-negative, this concludes that ${C^\star}\brc{\bc}\leq 0$. For the converse proof, we note that ${C^\star}\brc{\bc}  \leq 0$ implies $\sum_{k\in\setS} \tau_k C_k\brc{\bc} \leq 0$ for any $\tau_k\in\dbc{0,1}$~with $k\in\setS$. For any $j\in\setS$, the choice $\tau_j = 1$ and $\tau_k = 0$~for $k\neq j$ leads to $\sum_{k\in\setS} \tau_k C_k\brc{\bc} = C_j\brc{\bc}$. Thus, we can conclude that  ${C^\star}\brc{\bc} \leq 0$ implies $C_j\brc{\bc}\leq 0$. This concludes the proof.
	\end{proof}
	
	Invoking the above lemma, we now define $
		\bar{C}\brc{\bc} = \sum_{k\in\setS} \tau_k C_k\brc{\bc}$ as the \textit{overall penalty}. Our marginal problem hence reduces to finding a feasible point $\bc$ that satisfies $\bar{C}\brc{\bc} \leq 0$ and $\norm{\bc} = 1$ for any choice of $\tau_k\in\dbc{0,1}$ with $k \in \setS$. Using this alternative form of the feasible region, we now invoke Lemma~\ref{lem3} and find the point with minimal overall penalty, i.e., we find $\bc$ by solving the following optimization:
	\begin{align}
		\min_{\bc} \bar{C}\brc{\bc} \qquad \text{subject to} \; \norm{\bc} = 1. \label{optim:2}
	\end{align}
	This problem is of a standard form. To show this, let us represent the overall penalty  as 
	\begin{align}
		\bar{C}\brc{\bc} = \bc^\her \brc{\bar{\phi} \mI + \xi \mF \mT \mF^\her } \bc, \label{compact}
	\end{align}
	where $\bar{\phi}$ and $\xi$ are defined as
		$\bar{\phi} = \sum_{k\in\setS} \tau_k \phi_k^2$,
		$\xi =  \frac{P_{\rm u}  }{\sigma^2\brc{1+N/L}} \brc{\frac{N}{ L} \sum_{k\in\setS} \phi_k^2 - \epsilon_0 }$,
	and $\mT = \mathrm{diag} \set{\tau_1,\ldots,\tau_K}$ with $\tau_k = 0$ for $k\not\in\setS$. Given the compact form in \eqref{compact}, the solution to the optimization in \eqref{optim:2} is given by the eigenvector of $\bar{\phi} \mI + \xi \mF \mT \mF^\her $ that corresponds to its minimum eigenvalue: let $\lambda_1 \leq \ldots \leq \lambda_N$ and $\buu_1, \ldots, \buu_N$ denote the sorted eigenvalues of $\bar{\phi} \mI + \xi \mF \mT \mF^\her$ and their corresponding eigenvectors, respectively. The optimal receiver is then given by $\bc = \buu_1$. This solution can be readily calculated by means of \ac{svd}. To be consistent with Lemma~\ref{lem4}, the choice of $\tau_k$ in this solution needs to be the one that maximizes the overall penalty. In \cite{bereyhi2023device}, a low-complexity approach for choosing $\tau_k$ in each iteration has been proposed, which follows from the low sensitivity of the result on small deviation in $\tau_k$. We illustrate this scheme in the next section along with the selection metric.
  	\begin{algorithm}
		\caption{Joint Scheduling, Precoding, and Post-Processing}\label{alg:cap}
		\begin{algorithmic}
			\STATE Set $\setS = \dbc{K}$, choose  $\delta, \tau_0 \in \dbc{0,1}$, and find $\mW^\star$ by ~\eqref{eq:W_star} 
			\STATE Set $\mT = \mI_K$ and $\bc^\star = \buu_1$ with $\buu_1$ being the eigenvector of the largest eigenvalue of $\bar{\phi} \mI + \xi \mF \mT \mF^\her$
			\WHILE{$\mathrm{CRB} \brc{\setS, \mW^\star} > \Gamma_0$ or $\mae\brc{\bc^\star,\setS}  > \epsilon_0$}
			\STATE Find $k^\star = \argmax_{k\in\setS} \mas_k$ and update $\setS \leftarrow \setS - \set{k^\star}$ \texttt{\# MP Scheduling}
			\STATE Set $\tau_{k^\star}  = 0$ and update $\tau_k$ for $k\in\setS$ via \eqref{eq:cuttingset} \texttt{\# Subset-cutting weight update}
			\STATE Update $\mT \leftarrow \Diag{\tau_1, \ldots,\tau_K}$ and $\bc^\star \leftarrow \buu_1$.
			\ENDWHILE
		\end{algorithmic}
	\end{algorithm}
	\subsection{Selection Metric and the MP-based Algorithm}
	We now intend to find a metric by which we can choose the \textit{least effective} device. To this end, let us consider a particular iteration of the algorithm and assume that the scheduled set in this iteration is $\setS$. From the marginal problem $\mam_1$, we know that the designed precoding satisfies the power constraint, as well as the constraint on the downlink channel quality, i.e., constraints $(\mathrm{C}_4)$ and $(\mathrm{C}_5)$. The marginal problem $\mam_2$ further guarantees the satisfaction of $(\mathrm{C}_3)$. Our feasibility check is hence limited to validating $\mathrm{CRB} \brc{\setS, \mW^\star} \leq \Gamma_0$ and  $\mae \brc{\setS, \bc^\star} \leq \epsilon_0$, where $\mW^\star$ and $\bc^\star$ denote the solutions to the marginal problems $\mam_1$ and $\mam_2$, respectively. 
	
	We now assume that the derived solutions are infeasible. This means that either the \textit{sensing penalty}
	$
	C_{0} \brc{\setS,\mW^\star} = \mathrm{CRB} \brc{\setS, \mW^\star} - \Gamma_0
	$
	or the maximal aggregation penalty, i.e., $\max_{k \in \setS} C_k\brc{\bc^\star}$ for $C_k\brc{\bc^\star}$ defined in \eqref{eq:C_k}, is positive. We aim to find an index $k^\star\in\setS$ whose removal from $\setS$ leads to maximal reduction of the positive penalties in the next iteration. We address this task by defining the metric
 \begin{subequations}
	\begin{align}
		\mas_k &=  C_k\brc{\bc^\star} - \tau_0 {C_{0} \brc{\setS - \set{k},\mW^\star}},\\
		&= { \phi_k^2 }
		-  \frac{P_{\rm u} }{\sigma^2\brc{1+N/L}}  \brc{ \epsilon_0 - \frac{N}{ L} \sum_{k\in\setS} \phi_k^2  } {\abs{\bff_k^\her \bc^\star}^2 }, \nonumber \\ &- \tau_0 C_0 \brc{\setS - \set{k},\mW^\star},
	\end{align}
  \end{subequations}
	for some $\tau_0 \in\dbc{0,1}$. From the definition, one observes that $\mas_k$ describes the \textit{constraints violation} by client $k$: the larger $\mas_k$ is, the smaller the positive penalties will be after removing client $k$. We hence set $
		k^\star = \argmax_{k\in\setS} \mas_k, $
	which can be found by a linear search.
	\begin{remark}
		From the definition, one can observe that the maximization of $\mas_k$ is equivalent to the minimization of the inner product  $\abs{\bff_k^\her \bc^\star}$ added with some sensing-related penalty. We can hence perceive this selection metric as a regularized version of the classical matching pursuit~metric.
	\end{remark}
	
	The final algorithm is presented in Algorithm~1. 
	The proposed algorithm requires tuning penalty weights, i.e., $\tau_k$ for $k\in\dbc{K}$, as well as the regularization weight $\tau_0$. Given a fixed regularization weight $\tau_0$, the tuning of penalty weights can be addressed in each iteration via the \textit{subset-cutting} technique \cite{bereyhi2023device}. For device $k \in \setS$ in each iteration, we set $\tau_k$ to be 
	\begin{align}
		\tau_k = \begin{cases}
			\delta &\mas_k > 0\\
			1-\delta &\mas_k \leq 0,
		\end{cases}
	\label{eq:cuttingset}
	\end{align}
	for some $\delta\in\dbc{0,1}$. Noting that $\mas_k$ depends on $\tau_0$, we tune the weights by optimizing the performance of the algorithm against $\delta$ and $\tau_0$, which can be readily addressed numerically. The comparison against the optimally tuned $\tau_k$ in \cite{bereyhi2023device} shows that this approach closely follows the optimal solution as a drastically lower computation cost.
 	\begin{algorithm}
	\caption{Joint Design for OTA-FEEL and Sensing}
	\begin{algorithmic}[1]
	\FOR{iteration $t = 1, \ldots, T $}
		\STATE Choose active devices via \textbf{Algorithm 1} and store in $\setS$.
		\STATE Broadcast current model parameter $\mathbf{r}^{(t)}$ to all devices
		\FOR{each client $k\in \dbc{K}$ \textit{in parallel}}
            \STATE Set $\mathbf{r}_k^{(0)}=\mathbf{r}^{(t)}$
    		\STATE For $j=0,\ldots, \vartheta-1$, update $\mathbf{r}_k^{(j+1)}=\mathbf{r}_k^{(j)}-\zeta \nabla \mathcal{L}_k(\mathbf{r}_k^{(j)} \vert \mathcal{D}_k)$ and set $\mathbf{r}_k=\mathbf{r}_k^{(\vartheta)}$
		\ENDFOR
        \FOR{each client $k\in \setS$}
    		\STATE Transmit $s_k[m] = b_k[m] r_k[m]$ in $M$ intervals synchronously
		\ENDFOR
		\STATE Compute $\mathbf{R}$ as in (7) and decompose it as $\mathbf{R} = \mathbf{U} \boldsymbol{\Lambda} \mathbf{U}^\mathsf{H}$
		\STATE Set $\mathbf{T} = \mathbf{U} \boldsymbol{\Lambda}^{-1/2} \mathbf{U}^\mathsf{H}$, compute the sufficient statistics as $\tilde{\mathbf{y}}(\ell) = \mathbf{T} \mathbf{y}(\ell)$ for $\ell\in [L]$, and collect them in $\tilde{\mathbf{Y}}$
		\STATE Estimate echoes matrix as $\hat{\mathbf{G}}=\mathbf{T}^{-1} \tilde{\mathbf{Y}} {(\mathbf{W} \mathbf{D})}^\dagger$ and cancel its impact as $\bar{\mathbf{y}}(\ell) = \mathbf{y}(\ell) - \hat{\mathbf{G}} \mathbf{w}(\ell) x(\ell)$ for $\ell=1, \ldots, L$
		\STATE Compute global model from $\bar{\mathbf{y}}(\ell)$ and set $t\leftarrow t+1$
	\ENDFOR
	\end{algorithmic}
\end{algorithm}
\subsection{Complexity and Convergence of Scheduling Algorithm}
The proposed algorithm poses sub-quadratic complexity in terms of the number of devices: in iteration $t$ of the matching pursuit scheme, the algorithm performs a linear search over a subset of $K-t$ devices in $\setS$. To update the post-processing unit $\bc$ in \textit{each iteration}, the algorithm determines the eigenvector of $\bar{\phi} \mI + \xi \mF \mT \mF^\her $ that corresponds to its maximal eigenvalue. The direct approach is to determine the \ac{svd} of the matrix $\bar{\phi} \mI + \xi \mF \mT \mF^\her $, whose complexity scales with $KN^2$. Noting that the number of iterations is less than $K$ (after $K$ iterations, we end up with $\setS=\emptyset$), we conclude that with direct implementation via \ac{svd} the complexity of the algorithm scales with $K\cdot KN^2 = K^2N^2$. It is worth mentioning that the order of complexity is further reduced by using an extreme value decomposition technique in each iteration, e.g., the complexity of the iterative algorithm proposed in \cite{schwetlick2003iterative} scales with $K^qN^p$ for some $0 < q< 1$ and $1\leq p <2$.

As a comparing reference, one can consider random, greedy and exhaustive search algorithms: the former scales slower than the proposed algorithm by only a factor $K$, as it does not perform the linear search in each iteration. The greedy approach further scales similarly, as it performs linear search (by a different selection rule than the matching pursuit one) in each iteration. Both random and greedy search perform inferior to the proposed algorithm; see Section~\ref{sec:NumA} for numerical validation. The exhaustive search on the other hand scales exponentially with $K$ and hence is infeasible.


It is further easy to show that the proposed scheduling converges to a point which is locally optimal: considering the step-wise nature of matching pursuit technique \cite{foucart2013mathematical}, we can guarantee that the omitted device in each iteration of Algorithm~\ref{alg:cap} is the locally optimal choice, i.e., the one whose removal leads to maximal reduction in the aggregation error and Cram\'er-Rao bound. We further note that the marginal problems solved in Algorithm~\ref{alg:cap} are coupled only through $\setS$, and hence are solved optimally for the fixed $\setS$ in each iteration. This further guarantees the local optimally of $\mW^\star$ and $\bc^\star$. 
\vspace{-.4cm}
\subsection{Convergence of OTA-FEEL Framework}
We next provide convergence analysis for the \ac{ota-fl} framework, i.e., to discuss the convergence of federated averaging when employed via the proposed scheduling algorithm. To this end, we follow the classical assumptions that the individual point-wise loss function $\ell_i\brc{\br \vert \maD_k }$ for all $i$ and $k$ are \textit{continuously differentiable}  with respect to the model parameters. We further assume that the gradient of global loss $\mal\brc{\br_i \vert \maD}$ is Lipschitz continuous for all $\br_1, \br_2 \in \setR^D$ and some  $L_{Lip}>0$, i.e., $\norm{\nabla \mal\brc{\br_1\vert\maD} -  \nabla \mal\brc{\br_2\vert\maD}} \leq L_{Lip} \norm{\br_1-\br_2}$. For ease of analysis, we consider losses that are strongly convex with parameter $0 < \mu <L_{Lip}$, i.e., for all $\br, \boldsymbol{\delta} \in \setR^D$ 
\begin{align}
	\mal\brc{ \br + \boldsymbol{\delta}\vert \maD} \geq \mal\brc{\br\vert \maD} + \boldsymbol{\delta}^\trp \nabla \mal\brc{\br\vert \maD} + \frac{\mu}{2} \norm{\boldsymbol{\delta}}^2.
\end{align}
We further denote the minimal loss by $\mal^\star$. 



 We start the analysis by defining the \textit{optimality gap}:
\begin{definition}
	The optimality gap at communication round $t$ is the expected difference between the global loss computed at round $t$ and the minimal loss, i.e.,
	\begin{align}
		G_{t} = \Ex{\abs{\mal(\br^{(t)}\vert\maD) - \mal^\star } },
	\end{align}
where the expectation is taken with respect to the randomness of the learning algorithm, e.g., random batches in \ac{sgd}.
\end{definition}
Invoking Lemma 2.1. in \cite{friedlander2012hybrid}, we can bound the optimality gap in each iteration in terms of the last iteration as follows: 
\begin{align}
	G_{t+1} \leq \brc{1-\frac{\zeta}{L_{Lip}}} G_{t} + \frac{1}{2L_{Lip}} \mae^{(t)}\brc{\bc^\star, \setS}, \label{eq:upper}
\end{align}
where $\zeta$ is the learning rate, and $ \mae^{(t)}\brc{\bc^\star, \setS}$ denoting the aggregation error in communication round $t$ achieved by Algorithm~\ref{alg:cap}. Noting that the solution of Algorithm~\ref{alg:cap} is within the feasible region of the principal scheduling problem $\maP_1$, we can conclude $ \mae^{(t)}\brc{\bc^\star, \setS} \leq \epsilon_0$. Substituting into the upper bound, we have
\begin{align}
	G_{t+1} \leq \brc{1-\frac{\zeta}{L_{Lip}}} G_{t} + \frac{\epsilon_0}{2L_{Lip}}.\label{eq:upper2}
\end{align}
This implies that at any iteration, there exists a choice of learning rate which can reduce the optimality gap. This concludes the convergence of the \ac{ota-fl} setting. 
The converging optimailty gap is further proportional to the aggregation error bound $\epsilon_0$ which can be freely tuned in our proposed algorithm. 

\color{black}

	\section{Numerical Results and Discussions}
	In this section, we analyze the performance of OTA-FEEL with integrated sensing. We present the setup, numerical results, and insights on important network parameters.
	\subsection{Experimental Set-up and Benchmarks}
 \label{sec:NumA}
 We consider a setting with $K=20$ devices, $N=8$ antennas at the PS, and a single target, which is often located in an open area at a moderate distance from the server, or flying objects in the line of sight of the base station\cite{liu2022survey}. For the range of distances considered, we assume that the target is in the \ac{los} of the transmit array at the \ac{AoA} $\theta$ and distance $d$, with negligible scatterings from its surroundings\footnote{The reason to consider  LoS sensing channel is that the target is often located in open area in a moderate distance of the server, e.g., flying objects in the sight of the PSs  based on existing research works \cite{jing2024isac, liu2022survey, liu2022integrated, meng2023uav, lyu2022joint}, however, this is not a limitation of this framework and both sensing and communication channel models can be modeled as LoS or NLoS.}. In this 2-D setting, the sensing channel is given by $\mG = \alpha \mA\brc{\theta}$.
 Here, $\alpha$ describes the path-loss, and $\mA\brc{\theta} = \baa\brc{\theta} \baa\brc{\theta}^\trp$ for the array response $ \baa\brc{\theta} = \dbc{1, \exp\set{j2\pi d_{\rm e} \sin\theta}, \ldots, \exp\set{j2\pi \brc{N-1} d_{\rm e} \sin\theta}}^\trp$ with $d_{\rm e}$ being the antenna element spacing. \color{black} The path-loss $\alpha$ is determined as $\alpha = \varrho^\mathrm{t} \brc{d}=\varrho^\mathrm{\mathrm{t}}/{d}^{-\varepsilon_\mathrm{t}}$, where $\varrho^\mathrm{\mathrm{t}}$ is the path-loss at the reference distance $d_0=1$ m and $\varepsilon_\mathrm{t}$ denotes the path-loss exponent associated with the target echoes \cite{rappaport2015wideband}. Note that for sensing purposes, the signal travels the distance twice, i.e., radiation and reflection. Therefore, $\varepsilon_\mathrm{t}=4$ corresponds to the free-space scenario \cite{kumari2019adaptive}. It is worth mentioning that since our algorithm directly estimates $\mG$, assuming other settings does not impact the result. \color{black} To capture realistic environment characteristics, we set $\varepsilon_\mathrm{t}=4.5$. The reference intercept $\varrho^\mathrm{\mathrm{t}}$ is further determined by $\varrho^\mathrm{\mathrm{t}}=\dfrac{\sigma_{\rm RCS} \lambda^2}{\brc{4\pi}^3}$, where $\sigma_{\rm RCS}$ represents the target-radar cross-section and $\lambda$ is the wavelength. Throughout the simulations, we assume~$\sigma_{\rm RCS}=0.1$.
	 The devices are uniformly and randomly distributed around the \ac{ps} within a ring whose inner and outer radii are $D_{\rm in}=200$~m and $D_{\rm out}=500$~m, respectively. The maximum transmit powers of the \ac{ps} and edge devices are set to $P_{\mathrm{d}}=1$ W and $P_{\mathrm{u}}=1$ mW, respectively. The uplink and downlink channels are assumed to experience standard Rayleigh fading. This means that 
	$\bff_k=\varrho \brc{d_k}\bff_k^0$ and $\bh_k=\varrho \brc{d_k}\bh_k^0$, where $d_k$ represents the distance of device $k$ to the \ac{ps}. Here, $\bff_k^0$ and $\bh_k^0$ capture the small-scale fading effects and are generated as \ac{iid} complex-valued zero-mean Gaussian vectors with unit variance. Furthermore, $\varrho \brc{d_k}$ accounts for the path-loss determined as $\varrho \brc{d_k}=\varrho^0/{d_k}^{-\varepsilon_\mathrm{c}}$ with $\varepsilon_\mathrm{c}$ denoting the path-loss exponent specific to the device-to-\ac{ps} channel and $\varrho^0=\left( {\lambda}/{4 \pi }\right) ^2$. Throughout the simulations, we set $\varepsilon_\mathrm{c}=2.5$ and the carrier frequency to $5$ GHz. The noise power is further determined as $\sigma_{\rm PS}^2=N_0B$ for the bandwidth $B=10$ MHz and the noise power spectral density $N_0=9.8\times 10^{-18}$ W/Hz. In the simulations, the device scheduling strategy is adaptive, meaning that the selected devices do not have to be fixed for the entire FEEL process.
    We assume that~the selected devices are updated at the beginning of each coherence time interval during the CSI-acquisition phase. This approach ensures that the system adapts to changing channel conditions and optimizes the performance of OTA-FEEL over time.
 \begin{figure}[t!]
		\centering
		\includegraphics{ 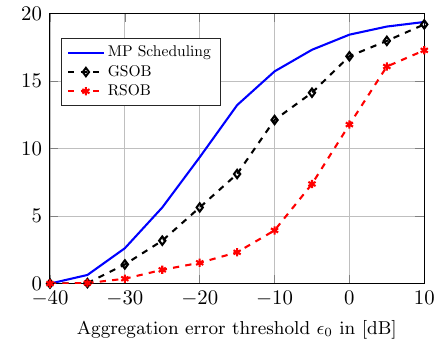} 
		\caption{ Average of selected devices versus aggregation error threshold.}
		\label{fig:Fig3a}
	\end{figure}
 \begin{figure}[t!]
	\centering
	\includegraphics{ 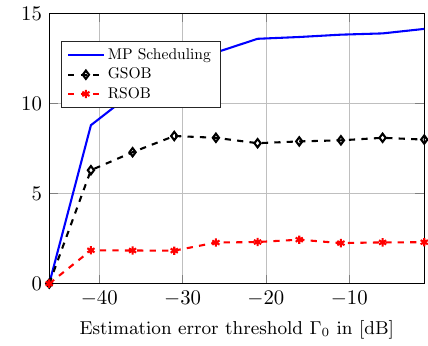} 
	\caption{Average number of selected devices versus sensing estimation error threshold.}
	\label{fig:Fig3b}
\end{figure}
\begin{figure}[t!]
	\centering
	\includegraphics{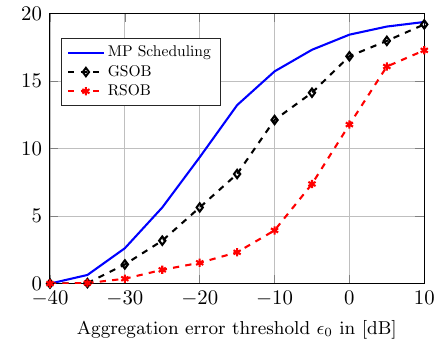} 
	\caption{\color{black}Average of selected devices versus aggregation error threshold.}
	\label{fig:FigF}
\end{figure}
    \begin{figure}[t!]
		\centering
		\includegraphics{ 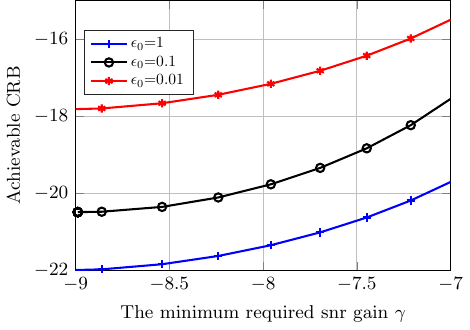} 
		\caption{Achievable error versus the minimum required downlink \ac{snr}. }
		\label{fig:Fig2}
	\end{figure}
 \begin{figure}[t!]
	\centering
	\includegraphics{ 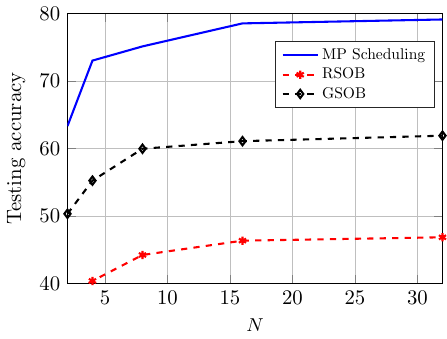} 
	\caption{  Testing accuracy versus the number of antennas at the PS.}
	\label{fig:FigNew}
\end{figure}
  \begin{figure}[t!]
	\centering
	\includegraphics{ 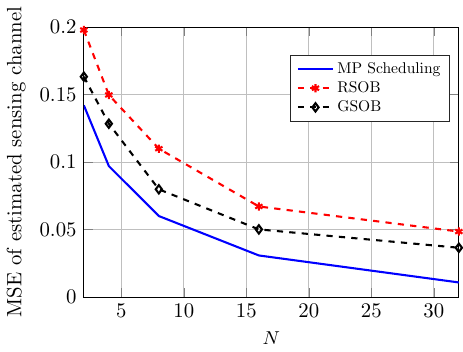} 
	\caption{ MSE of the estimated sensing channel versus $N$.}
	\label{fig:FigNew2}
\end{figure}
 \par For comparison, we also consider the following approaches:
	\begin{itemize}[leftmargin=0.25cm]
	\item \textit{Random Scheduling with Optimal Beamforming (RSOB)}: During each iteration, we randomly delete one user and design $\mW$ and $\bc$ based on \eqref{eq:W_star} and \eqref{compact},~respectively. 
	\item \textit{Greedy Scheduling with Optimal Beamforming~(GSOB)}:
In each iteration, we remove the device~with weakest channel and design $\mW$ and $\bc$ based on \eqref{eq:W_star} and \eqref{compact},~respectively.
	\end{itemize}
 
\begin{figure*}[t!]
\begin{tabular*}{\textwidth}{@{\extracolsep{\fill}}cccc} 
	\includegraphics{ 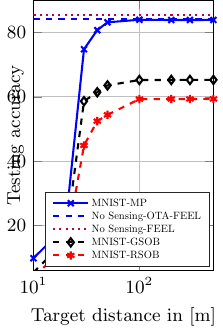} 
	&	\includegraphics{ 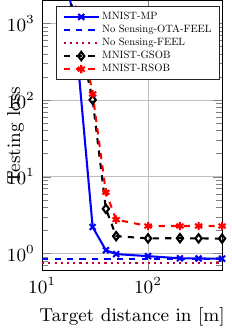} 
	&	\includegraphics{ 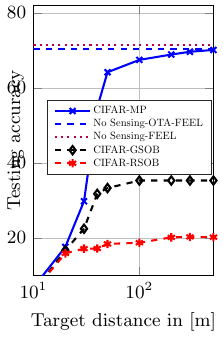}  
	&	\includegraphics{ 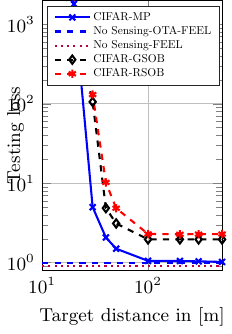}  \\
	 (a) & (b) & (c) & (d)\\
\end{tabular*}
\caption{ Testing accuracy and testing loss versus target distance for MNIST and CIFAR dataset.}
\label{Fig8}
\end{figure*} 
Next, we check if the designed precoder and receiver lead to a feasible point in the original problem. If infeasibility arises, we delete the user again given the aforementioned criteria.	
	\subsection{Learning Setup} 
 	We consider two learning setups: in the first setup, we use logistic regression on the MNIST dataset, which has 60,000 $28\times 28$ images. We train on 83\% of the dataset and test on the remaining portion \cite{xiao2017fashion}. The training and testing datasets are randomly shuffled and unevenly distributed among devices. We assign random percentages to each client using the Dirichlet distribution, simulating non-uniform data distribution in federated learning. This introduces diversity in data distribution, which is important for simulating federated learning scenarios.
  In addition, we perform image classification on the CIFAR-10 dataset \cite{krizhevsky2009learning}. It contains the same number of images, with each image having 3 color channels and dimensions of $32\times32$ pixels. The dataset is divided into 10 classes. We use the same training and testing split. 
For this task, we employ a deep CNN architecture consisting of three convolutional blocks, each including two convolutional layers, where the first convolutional layer in each block uses $3 \times 3$ kernels with a varying number of filters (32, 64, 128, and 256), followed by ReLU activation and a $2 \times 2$ max-pooling layer. The convolutional layers are designed to progressively capture more abstract features from the input images.

Following the convolutional blocks, the network features three fully connected layers. The first two fully connected layers have 1024 and 512 units respectively, with ReLU activations applied after each layer. The final fully connected layer outputs a 10-dimensional vector, representing the class probabilities for classification. We note that while our experiments primarily focus on the MNIST and CIFAR-10 datasets, the proposed framework can be easily extended to more complex datasets, such as FMNISt or ImageNet with similar insights expected.
\vspace{-.4cm}
\subsection{Results and Discussions}
	\subsubsection{Scheduling Performance}
	\par To evaluate the performance of the proposed matching pursuit-based  algorithm, the average number of selected devices is plotted against aggregation error threshold $\epsilon_0$ in Fig. \ref{fig:Fig3a}. The proposed algorithm significantly outperforms the considered benchmarks. The same trend can be observed in Fig. \ref{fig:Fig3b}, which displays the average number of selected devices against the sensing estimation error threshold $\Gamma_0$.
	The results demonstrate that the proposed approach based on matching pursuit is able to select more devices than the other state-of-the-art approaches. Another intuitive finding is given by observing the trend in both figures: by loosening the restriction on the sensing and aggregation error, a larger number of users will participate in FEEL learning process. Comparing the two figures further indicates that the increase in the aggregation error threshold $\epsilon_0$ leads to faster growth in the number of selected devices as compared to the sensing error threshold $\Gamma_0$. This observation is intuitive, since the scheduling is primarily controlled by the threshold on the aggregation error while the sensing error threshold implicitly influences the scheduling through its effect on the effective noise variance. 
 Moreover, the average number of selected devices versus the aggregation error threshold $\epsilon_0$ for different values of $K$ is plotted in Fig.~\ref{fig:FigF}. Our simulation results show that although the dataset size per user decreases as the number of users increases, the total dataset in each round generally increases. This indicates that the system can accommodate a higher number of users while maintaining a substantial total dataset per round, which could potentially enhance learning performance.

 \subsubsection{Sensing-Learning-Communication Trade-off}
Fig. \ref{fig:Fig2} illustrates a three-dimensional trade-off between sensing, learning and communication. The figure shows the achievable Cram\'er-Rao bound against the minimum \ac{snr} scaling $\gamma$ in the downlink channel for various choices of the tolerable aggregation error $\epsilon_0$. As the figure shows, for a given $\epsilon_0$, the achievable Cram\'er-Rao bound increases as $\gamma$ grows large. This is a logical trade-off: as we increase $\gamma$, the downlink beam focuses toward the edge devices for downlink communication, thereby reducing its focus on the target for sensing. This leads to diminished sensing quality and consequently, an increased sensing error. The figure further depicts the learning and sensing trade-off: as $\epsilon_0$ increases, the achievable Cram\'er-Rao bound drops. This is logically consistent, since by increasing $\epsilon_0$, we loosen the constraint on aggregation error. This means that the scheduling mainly focuses on reducing the sensing error, which in turn leads to a decrease in the sensing error. 

Figs.~\ref{fig:FigNew} and \ref{fig:FigNew2} demonstrate the performance of the proposed algorithm against \ac{ps} array-size. As the figure shows, the proposed algorithm depict the same trend as the benchmarks with a marginally large gap. The gap is in particular large for testing accuracy. This is rather intuitive, since the proposed scheme follows a learning-aware metric for scheduling whereas the benchmarks schedule only based on the quality of the communication channels.

 \subsubsection{Learning Performance}
 Figs.~\ref{Fig8}~(a)-(d) illustrate the test accuracy and loss of classification over the MNIST and CIFAR datasets against the distance $d$~of~the~target from the \ac{ps} under the assumption that the clients remain at fixed distances. As a comparison, we consider a scenario without sensing at the PS for both OTA-FEEL and standard FEEL cases. In this case, the PS solely focuses on model aggregation without any involvement in target sensing. As depicted in the figures, the learning algorithm fails to converge when $d$ is chosen to be very small. This is due to the fact that at such proximate distances, the sensing echos strongly interfere, and their impact is not suppressed even after cancellation. The learning algorithm starts to converge and rapidly reach saturation as the target moves towards moderate distances. In this range, the angle estimation is still accurate, and hence, the interference after cancellation weakly impacts model aggregation. At very large distances, the sensing signal becomes so faint that its influence can be neglected. This can be observed by comparing the performance with the sensing-free scenarios. The results for the proposed algorithm are further compared against the benchmarks depicting its superiority.

  \subsubsection{Sensing Performance}
  Fig.~\ref{fig:Fig7}~(a) shows the \ac{mse} of~the estimated sensing channel, denoted as $\mathrm{MSE} = \Ex{\norm{\mG - \hat{\mG}}^2 }{ }$, plotted against the distance $d$ of the target from the \ac{ps}. This analysis assumes that the clients remain at fixed distances for the MNIST dataset. The observed behavior is consistent with the one seen in Fig.~\ref{Fig8}:~ for small values of $d$, the sensing signal is strong, resulting in accurate estimations. However, as the target moves towards far distances, the sensing signal weakens, leading to a decline in sensing performance. This trend is further illustrated in Fig. \ref{fig:Fig7}~(b), which demonstrates that the estimated target's \acp{AoA} closely match the true \acp{AoA}. 
  At larger distances (approximately for $d>100$ m), the target's movement away from the PS leads to a significant increase in \ac{mse}, making it challenging for the PS to accurately estimate the target's \ac{AoA}. This phenomenon is depicted in Fig.~\ref{fig:Fig7}~(c).
This result suggests that there is always a specific range of distances at which both the desired sensing and learning behaviors are satisfied. Generally, this range depends on various parameters, such as the maximum allowed aggregation error and the bound on the sensing error metric (i.e., CRB). A rigorous investigation of these distance ranges under a set of abstract assumptions is, however, an interesting avenue for further research.

\vspace{-2mm}
\section{Conclusion} 
 This paper presents a novel framework that effectively integrates OTA-FEEL with sensing, providing a balanced approach to address the dual functionalities of wireless networks. The system design leverages the downlink transmission as a dual-purpose signal, serving as a carrier for the aggregated model and a radar probing pulse. By incorporating a single beamformer, we efficiently fulfill the communication needs and enable target sensing through maximum likelihood estimation of the reflected signals. This design ensures the quality of downlink transmission and estimation accuracy for sensing and restricts the aggregation error incurred during model aggregation. Moreover, to optimize the system's performance, we propose a Matching Pursuit-based user scheduling strategy that dynamically prioritizes either sensing or model aggregation based on real-time requirements. This flexibility allows the system to achieve a balance between sensing accuracy, communication efficiency, and computation needs. Extensive numerical investigations reveal scenarios where sensing can be achieved "for free." This work focuses on scenarios with a stationary target. Extending the
approach to moving targets is an interesting direction for further research. 
\vspace{-3mm}
	\appendices
	\section{ } 
	Let us assume that $\rmr_k\dbc{\ell}$ is distributed such that $\mF \mB_\setS \br \dbc{\ell}$ is distributed Gaussian. 
	Note that such an assumption holds, if either the local gradients are Gaussian distributed  or the number of selected devices is large enough such that the central limit theorem holds. It is worth mentioning that this is in general a realistic assumption as discussed in \cite{lee2020bayesian}.
	Moreover, since $\bzz \dbc{\ell} \sim \mathcal{CN}(0, \sigma_{\rm PS}^2 \mI)$, we can conclude that $\tilde{\bzz} \dbc{\ell}$ is Gaussian, i.e., $\tilde{\bzz} \dbc{\ell} \sim \mathcal{CN}\brc{ \boldsymbol{0}, \mR }$ where 
	\begin{align}
		&\mR =   \Ex{\tilde{\bzz} \dbc{\ell} \tilde{\bzz} \dbc{\ell}^\her }{ },\nonumber\\ & =  \Ex{ (\beta_{\rm c} \dbc{\ell} \mF \mB_\setS  \br \dbc{\ell} + \bzz \dbc{\ell})(\beta_{\rm c} \dbc{\ell}   \br^\her \dbc{\ell}\mB_\setS^\her \mF^\her + \bzz \dbc{\ell}^\her)}{ },\nonumber\\
		&=\beta_{\rm c} \dbc{\ell} \Ex{\mF \mB_\setS  \br \dbc{\ell}\br^\her \dbc{\ell}\mB_\setS^\her \mF^\her}{ }\nonumber\\& \; \; \; +\beta_{\rm c} \dbc{\ell} \Ex{\mF \mB_\setS \br \dbc{\ell}\bzz \dbc{\ell}^\her}{} + \beta_{\rm c} \dbc{\ell}\Ex{\bzz \dbc{\ell}    \br^\her \dbc{\ell}\mB_\setS^\her \mF^\her}{}\nonumber\\ & \; \; \; + \Ex{\bzz \dbc{\ell} \bzz \dbc{\ell}^\her}{}.\nonumber
	\end{align}
	Since $\bzz \dbc{\ell} $ is zero-mean, and noise and signal are independent, the second and third terms are zero.  By considering the fact that $\br\dbc{\ell}$ has uncorrelated and unit-variance entries, we conclude that
	$
	\mR= \beta_{\rm c}  \dbc{\ell} \mF \mB_\setS \mB_\setS ^\her \mF^\her  + \sigma_{\rm PS}^2 \mI.
	$
  \begin{figure}
\begin{center}
  \includegraphics{ 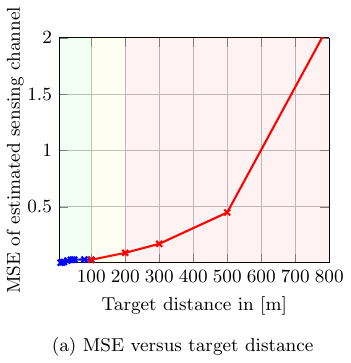}
   \\ \vspace{.18cm}
\includegraphics{ 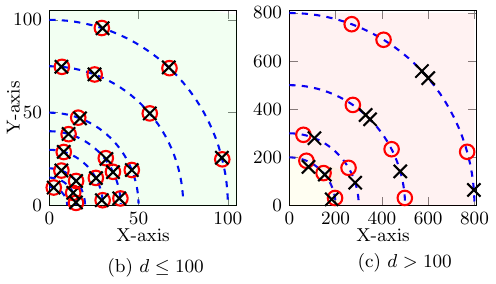} 
\end{center}
\centering
\caption{ The sensing performance of OTA-FEEL integrated~sensing.}
\begin{tabular}{r@{: }l@{\quad}r@{: }l}
   \hspace{-.3cm} {\Large {$\boldsymbol{\times}$}} & \small Estimated target location &  {\Large {$\textrm{o}$}} & \small Ground truth target location 
\end{tabular}
\label{fig:Fig7}
\end{figure}
 \vspace{-3mm}
	\section{} We now consider the deterministic unknown matrix $\mG$ and write its likelihood function  as
		$\mal\brc{\mG} = \prod_{\ell=\ell_0}^{\ell_0+L} p\brc{\tilde{\byy}\dbc{\ell} \vert \mG}$,
	where
	$p\brc{\tilde{\byy}\dbc{\ell} \vert \mG} = \pi^{-N} \exp\set{-\norm{\tilde{\byy}\dbc{\ell}  - \mT\mG \bww \dbc{\ell} x \dbc{\ell}}^2 }.$ The log-likelihood function is
	\begin{align}
		\hspace{-.31cm}\log\mal\brc{\mG} = -LN\log \pi - \sum_{\ell=1}^{L} \norm{\tilde{\byy}\dbc{\ell}  - \mT\mG 
			\bww \dbc{\ell} x \dbc{\ell}}^2. \label{eq:logliklihood}
	\end{align}
	We now use the alternative representation by the Kronecker product.  Since $\mT\mG \bww \dbc{\ell} $ is a column vector, it equals the vectorized form of its transpose, i.e., 
	$\mT\mG \bww \dbc{\ell} =   \mathrm{vec} \brc{\brc{\mT\mG \bww \dbc{\ell}}^\trp}=\mathrm{vec} \brc{\bww ^\trp\dbc{\ell}\mG^\trp \mT^\trp }\stackrel{\star}{=}(\mT \otimes \bww^\trp \dbc{\ell}) \bvv$, where $\bvv = \mathrm{vec} \brc{\mG}^\trp$ is the $N^2$-dimensional vectorized form of $\mG$ and $\star$ comes from the identity $\mathrm{vec} \brc{\mA \mB \mC}= (\mC^\trp \otimes \mA) \mathrm{vec}\brc{\mB}$. The log-likelihood term is hence written as
	\begin{align}
		\log\mal\brc{\mG} = -LN\log \pi - \sum_{\ell=1}^{L} \norm{\tilde{\byy}\dbc{\ell}  - x \dbc{\ell}\brc{ \mT \otimes \bww^\trp \dbc{\ell}}  \bvv}^2. \nonumber
	\end{align}
	We now determine the matrix of Fisher information as
	\begin{align*}
	    I\brc{\mG} = - \Ex{\nabla^2 \log\mal\brc{\mG} }{ \prod_{\ell=1}^{L} p\brc{\tilde{\byy}\dbc{\ell} \vert \mG} }.
	\end{align*}
	In general, the derivatives are taken with respect to the augmented form~of~$\bvv$. We can however invoke the quadratic form of the log-likelihood and take the derivatives directly in the complex space. The Fisher information matrix is given by
	\begin{align}
		I\brc{\mG} = 2\sum_{\ell=1}^{L} \abs{x \dbc{\ell}}^2   \brc{ \mT^\her\otimes \bww^*\dbc{\ell} } \brc{ \mT \otimes \bww^\trp\dbc{\ell} },
	\end{align}
	and the standard Fisher information matrix is determined by augmenting this matrix. 
	Therefore, using $\mT^\her \mT  = \mR^{-1}$~and the identity $\brc{\mA \otimes \mB} \brc{\mC \otimes \mD} = {\mA\mC \otimes \mB\mD}$ we can write
	$I\brc{\mG} =  2\sum_{\ell=1}^{L}  \mR^{-1} \otimes \bww^*\dbc{\ell} \bww^\trp\dbc{\ell}$.
	The Cram\'er-Rao bound indicates that the covariance matrix of the error~term $\hat{\bvv} - \bvv$ for the \ac{ml} estimation $\hat{\bvv}$ converges to
	$\mC = \frac{1}{2} \brc{ \mR^{-1} \otimes\sum_{\ell=1}^{L}  \bww^*\dbc{\ell} \bww^\trp\dbc{\ell}}^{-1}
	$ as $L$ grows large.~By~using $\brc{\mA \otimes \mB}^{-1} = {\mA^{-1} \otimes \mB^{-1}}$, the Cram\'er-Rao bound reads $
		\mC = \frac{1}{2} \brc{\mR^{-1} \otimes \mW^* \mW^\trp }^{-1}
		= \frac{1}{2} \mR \otimes \brc{\mW^* \mW^\trp }^{-1}.
$
\vspace{-3mm}
	\section{}
	\par Following the Cram\'er-Rao theorem, we can conclude that for large $L$ and $M$, the error term $\rme \dbc{\ell} $ is a zero-mean Gaussian random variable. Since the first term and second term of $\rme \dbc{\ell}$ are independent, the variance is given by
	\begin{subequations}
		\begin{align}
			&\Ex{\abs{\rme \dbc{\ell}}^2}{} = \beta_{\rm s} \dbc{\ell} \Ex{\abs{\frac{\bc^\her  }{\sqrt{\eta}}  \brc{\mG - \hat{\mG}}
					\bww \dbc{\ell} }^2 }{}\nonumber\\&\; \; \;+ \Ex{\abs{ \frac{\bc^\her  }{\sqrt{\eta}}  \bzz \dbc{\ell}}^2}{} ,\nonumber\\
			&=\dfrac{ \beta_{\rm s}\dbc{\ell}}{\eta} \bc^\her \Ex{\brc{\mG - \hat{\mG}} \bww \dbc{\ell} \bww ^\her \dbc{\ell} \brc{\mG - \hat{\mG}}^\her}{} \bc\nonumber\\&\; \; \;+ \dfrac{ \norm{\bc}^2}{\eta} \sigma_{\rm PS}^2.\nonumber
		\end{align} 
	\end{subequations}
	By employing the alternative representation using the Kronecker product, i.e., $(\mG-\hat{\mG}) \bww \dbc{\ell} =  (\mI \otimes \bww^\trp \dbc{\ell}) (\bvv-\hat{\bvv})$ with $\bvv$ and $\hat{\bvv}$ being the vectorized form of $\mG$ and $\hat{\mG}$, we have 
	\begin{subequations}
		\begin{align}
			&\Ex{\abs{\rme \dbc{\ell}}^2}{} \nonumber\\&=\dfrac{\beta_{\rm s}\dbc{\ell}}{\eta}  \bc^\her \Ex{(\mI \otimes \bww^\trp \dbc{\ell}) (\bvv-\hat{\bvv})  (\bvv-\hat{\bvv}) ^\her (\mI \otimes \bww^\trp \dbc{\ell})^\her }{} \bc\nonumber \\ &\; \; \; +  \dfrac{\norm{\bc}^2}{\eta} \sigma_{\rm PS}^2,\nonumber\\
			&=\dfrac{\beta_{\rm s} \dbc{\ell}}{\eta} \bc^\her (\mI \otimes \bww^\trp \dbc{\ell}) \Ex{ (\bvv-\hat{\bvv})  (\bvv-\hat{\bvv}) ^\her }{} (\mI \otimes \bww^* \dbc{\ell})\bc \nonumber\\& \; \; \; +\dfrac{\norm{\bc}^2}{\eta} \sigma_{\rm PS}^2,\nonumber \\ &=\dfrac{ \beta_{\rm s} \dbc{\ell}}{\eta} \brc{\bc^\her\otimes \bww^\trp \dbc{\ell}} \mC \brc{\bc \otimes \bww^* \dbc{\ell}}
			+  \dfrac{ \norm{\bc}^2 }{\eta}\sigma_{\rm PS}^2. \nonumber
		\end{align} 
	\end{subequations}

\IEEEpeerreviewmaketitle
\bibliography{ref}
\bibliographystyle{IEEEtran}
\vspace{-15mm}
\begin{IEEEbiography}
[{\includegraphics[width=1.1in,height=1.2in,clip,keepaspectratio]{ 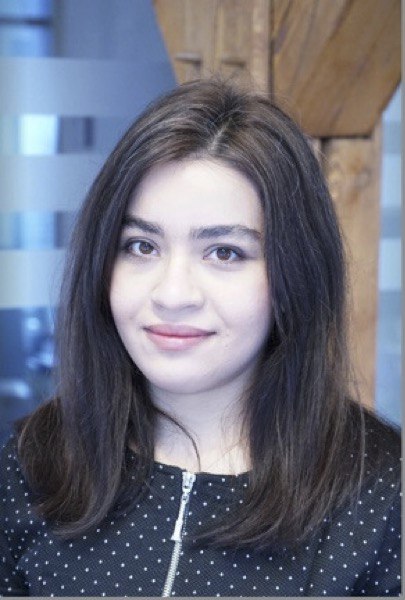}}]
{Saba Asaad} (Member IEEE) received the B.Sc. and
M.Sc. degrees in electrical engineering, communi-
cations from the Sharif University of Technology,
Tehran, Iran, in 2012 and 2014, respectively, and the Ph.D. degree (with distinction) from the University of Tehran, Tehran, Iran, in 2018. From 2018 to 2023, she was with the Institute for Digital Communications, Friedrich-Alexander-University of Erlangen-Nuremberg, Erlangen, Germany, as a Postdoctoral Fellow. Since 2023, she has been with the Next Generation Wireless Networks Research Lab, York University, Toronto, ON, Canada, as a Postdoctoral Research Associate. Her research interests include distributed learning over wireless networks, information theory, wireless communications with focus on MIMO transmission techniques, and physical layer security.
\end{IEEEbiography}
\vspace{-15mm}
\begin{IEEEbiography}
[{\includegraphics[width=1.1in,height=1.4in,clip,keepaspectratio]{ 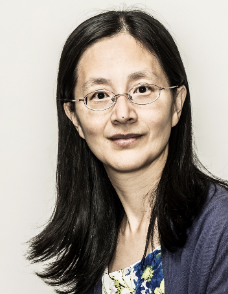}}]
{Ping Wang} (Fellow, IEEE) is a Professor at the Department of Electrical Engineering and Computer Science, York University, and a Tier 2 York Research Chair. Prior to that, she was with Nanyang Technological University, Singapore, from 2008 to 2018. Her recent research interests focus on integrating Artificial Intelligence (AI) techniques into communications networks. Her scholarly works have been widely disseminated through top-ranked IEEE journals/conferences and received the IEEE Communications Society Best Survey Paper Award in 2023, and the Best Paper Awards from IEEE prestigious conference WCNC in 2012, 2020 and 2022, from IEEE Communication Society: Green Communications \& Computing Technical Committee in 2018, from IEEE flagship conference ICC in 2007. She has been serving as the associate editor-in-chief for IEEE Communications Surveys \& Tutorials and an editor for several reputed journals, including IEEE Transactions on Wireless Communications. She is a Distinguished Lecturer of the IEEE Vehicular Technology Society (VTS). She is also the Chair of the Education Committee of IEEE VTS.
\end{IEEEbiography}
\vspace{-15mm}
\begin{IEEEbiography}
[{\includegraphics[width=1.1in,height=1.4in,clip,keepaspectratio]{ 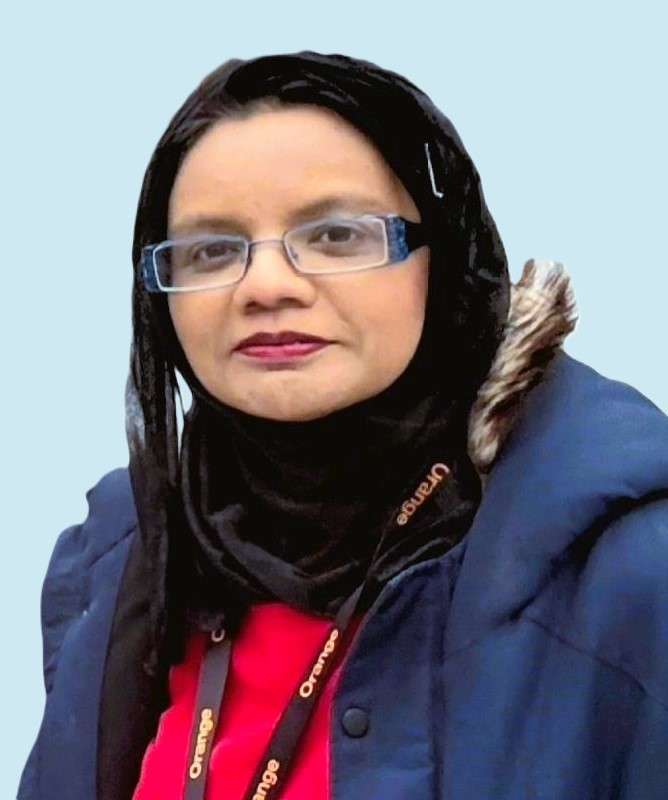}}]
{Hina Tabassum}  (Senior Member, IEEE) received the Ph.D. degree from the King Abdullah University of Science and Technology (KAUST). She is currently an Associate Professor with the Lassonde School of Engineering, York University, Canada, where she joined as an Assistant Professor, in 2018. She is also appointed as a Visiting Faculty at University of Toronto in 2024 and the York Research Chair of 5G/6G-enabled mobility and sensing applications in 2023, for five years. Prior to that, she was a postdoctoral research associate at University of Manitoba, Canada.  She has been selected as IEEE ComSoc Distinguished Lecturer (2025-2026). She is listed in the Stanford's list of the World’s Top Two-Percent Researchers in 2021-2024. She received the Lassonde Innovation Early-Career Researcher Award in 2023 and the N2Women: Rising Stars in Computer Networking and Communications in 2022.  She has been recognized as an Exemplary Editor by the IEEE Communications Letters (2020), IEEE Open Journal of the Communications Society (IEEE OJCOMS) (2023 - 2024), and IEEE Transactions on Green Communications and Networking (2023).  She was recognized as an Exemplary Reviewer (Top 2\% of all reviewers) by IEEE Transactions on Communications in 2015, 2016, 2017, 2019, and 2020. She is the Founding Chair of the Special Interest Group on THz communications in IEEE Communications Society (ComSoc)-Radio Communications Committee (RCC). She served as an Associate Editor for IEEE Communications Letters (2019–2023), IEEE OJCOMS (2019–2023), and IEEE Transactions on Green Communications and Networking (2020–2023). Currently, she is also serving as an Area Editor for IEEE OJCOMS and an Associate Editor for IEEE Transactions on Communications, IEEE Transactions on Wireless Communications, and IEEE Communications Surveys and Tutorials.
\end{IEEEbiography}
\end{document}